\documentclass[11pt,letterpaper]{article}
\usepackage[utf8]{inputenc}
\usepackage{amsmath,amssymb,amsthm,mathrsfs,amsfonts}
\usepackage{bm}
\usepackage{color}
\usepackage{enumerate}
\usepackage{physics}
\usepackage{complexity}
\usepackage{soul,xcolor}
\usepackage{textcomp}
\usepackage{verbatim}
\usepackage{tikz}
\usetikzlibrary{arrows}
\usetikzlibrary{backgrounds}
\usetikzlibrary{quantikz}
\usepackage{caption}
\usepackage{subcaption}
\usepackage{algpseudocode}
\usepackage{algorithm}
\usepackage[backend=biber,style=alphabetic,url=false,isbn=false,maxnames=10,minnames=3,minalphanames=3]{biblatex}
\newtheorem{corollary}{Corollary}
\newtheorem{definition}{Definition}
\newtheorem{lemma}{Lemma}
\newtheorem{remark}{Remark}

\newtheorem{theorem}{Theorem}
\newtheorem{problem}{Problem}

\newcommand{\mc}{\mathcal}
\renewcommand{\E}{\mathop{\mathbb E\/}}

\newcommand{\vspan}{\mathrm{span}}
\newclass{\sharpp}{\#P}
\newclass{\cocequalp}{coC_{=}P}
\newfunc{\expp}{exp}

\newcommand{\ssym}[2]{\vee^{#2}\mathbb{C}^{#1}}
\newcommand{\psym}[2]{\Pi_{\mathrm{sym}}^{#1,#2}}
\newcommand{\Var}{\mathrm{Var}}
\newcommand{\swapop}{\mathrm{SWAP}}
\newcommand{\mmix}{\sigma_{\mathrm{m}}}
\newcommand{\ipe}{\mathrm{IP}_{\varepsilon}}

\usepackage[colorlinks = true]{hyperref}
\usepackage{longtable}

\usepackage{xcolor}
\definecolor{darkred}  {rgb}{0.5,0,0}
\definecolor{darkblue} {rgb}{0,0,0.5}
\definecolor{darkgreen}{rgb}{0,0.5,0}

\hypersetup{
  urlcolor   = blue,         
  linkcolor  = darkblue,     
  citecolor  = darkgreen,    
  filecolor   = darkred       
}

\usepackage[
letterpaper,
top=1in,
bottom=1in,
left=1in,
right=1in]{geometry}
\newtheorem*{theorem*}{Theorem}

\begin{document}
\title{Distributed quantum inner product estimation}
\author{
Anurag Anshu\thanks{Department of EECS, Challenge Institute for Quantum Computation and Simons Institute for the Theory of Computing, UC Berkeley. \href{mailto:anuraganshu@berkeley.edu}{anuraganshu@berkeley.edu}}
\and Zeph Landau\thanks{Department of EECS, UC Berkeley. \href{mailto:zeph.landau@gmail.com}{zeph.landau@gmail.com}}
\and Yunchao Liu\thanks{Department of EECS, UC Berkeley. \href{mailto:yunchaoliu@berkeley.edu}{yunchaoliu@berkeley.edu}}
}

\date{}
\maketitle

\begin{abstract}
As small quantum computers are becoming available on different physical platforms, a benchmarking task known as cross-platform verification has been proposed that aims to estimate the fidelity of states prepared on two quantum computers. This task is fundamentally distributed, as no quantum communication can be performed between the two physical platforms due to hardware constraints, which prohibits a joint SWAP test. In this paper we settle the sample complexity of this task across all measurement and communication settings. The essence of the task, which we call distributed quantum inner product estimation, involves two players Alice and Bob who have $k$ copies of unknown states $\rho,\sigma$ (acting on $\mathbb{C}^{d}$) respectively. Their goal is to estimate $\mathrm{Tr}(\rho\sigma)$ up to additive error $\varepsilon\in(0,1)$, using local quantum operations and classical communication. In the weakest setting where only non-adaptive single-copy measurements and simultaneous message passing are allowed, we show that $k=O(\max\{1/\varepsilon^2,\sqrt{d}/\varepsilon\})$ copies suffice. This achieves a savings compared to full tomography which takes $\Omega(d^3)$ copies with single-copy measurements. Surprisingly, we also show that the sample complexity must be at least $\Omega(\max\{1/\varepsilon^2,\sqrt{d}/\varepsilon\})$, even in the strongest setting where adaptive multi-copy measurements and arbitrary rounds of communication are allowed. This shows that the success achieved by shadow tomography, for sample-efficiently learning the properties of a single system, cannot be generalized to the distributed setting. Furthermore, the fact that the sample complexity remains the same with single and multi-copy measurements contrasts with single system quantum property testing, which often demonstrate exponential separations in sample complexity with single and multi-copy measurements. 
\end{abstract}

\section{Introduction}
We consider the following task.  Alice is given $k$ copies of an unknown state $\rho$ acting on $\mathbb{C}^{d}$ and Bob is given $k$ copies of a possibly different unknown state $\sigma$ acting on $\mathbb{C}^{d}$. Their goal is to estimate $\Tr(\rho\sigma)$ up to additive error $\varepsilon\in(0,1)$ with success probability at least $2/3$, using local quantum operations and classical communication (LOCC). We refer to this task as {\it distributed quantum inner product estimation}.

We stress that the distributed aspect of the task, namely that there is no quantum channel between Alice and Bob, is fundamental to the problem.  If we were to allow a quantum channel, the task would be easily accomplished with $k= O(1/\varepsilon^2)$ copies, by Alice sending Bob copies of her state and Bob performing the SWAP test with his copies.  Distributed quantum inner product estimation is the central step of a procedure known as cross-platform verification, proposed in~\cite{Elben2020cross}: the goal is to test whether the output states of two quantum computers on different platforms (such as ion traps and superconducting qubits) are close to each other, but joint quantum operation  across the platforms cannot be performed. A related application is to use a well-calibrated quantum computer to test another quantum computer on a different platform. See Section~\ref{sec:discussion} for a discussion of other applications.

Within our distributed setting without quantum communication there are still a number of communication features to consider, namely what kind of classical communication shall be allowed between Alice and Bob.  The weakest would be to simply allow Alice and Bob to send information to a third party that then performs a classical calculation (known as simultaneous message passing, Fig.~\ref{fig:communication}a). The strongest would be to allow multi-round communication between Alice and Bob (Fig.~\ref{fig:communication}c). In addition, we also consider in what form Alice and Bob can access their samples: the simplest setting would be access to a single copy of the state at a time; we also consider the most general setting where Alice and Bob can access multiple copies of their state at a time, allowing them to do a joint measurement on multiple copies in tensor product. Furthermore, their measurement bases can be \emph{adaptive}: each new measurement operator could depend on all previous measurement outcomes and communication.

In this paper we settle the sample complexity of distributed quantum inner product estimation under all these possible specifications of communication and sample access. We show that the sample complexity remains the same (up to a constant factor) across all measurement and communication settings, ranging from the weakest non-adaptive single-copy measurements and simultaneous message passing (Fig.~\ref{fig:communication}a), to the strongest adaptive multi-copy measurements and arbitrary rounds of classical communication (Fig.~\ref{fig:communication}c). 

\begin{theorem*} [Main result] There exists an algorithm for inner product estimation with non-adaptive single-copy measurements and simultaneous message passing that uses $k=O(\max\{1/\varepsilon^2,\sqrt{d}/\varepsilon\})$ copies. Furthermore, any algorithm for inner product estimation with arbitrary rounds of adaptive multi-copy measurements and classical communication must use at least $\Omega(\max\{1/\varepsilon^2,\sqrt{d}/\varepsilon\})$ copies.
\end{theorem*}

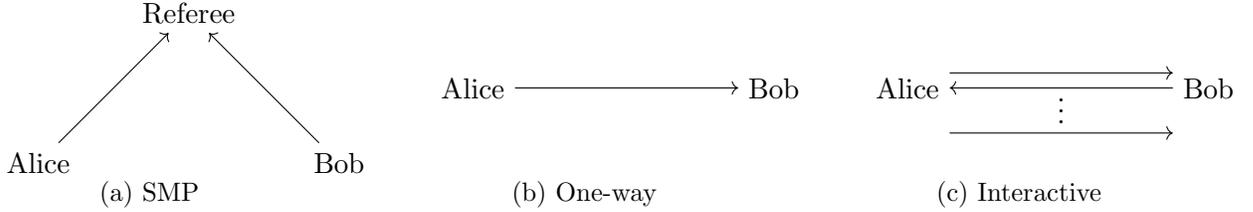
\begin{figure}[t]
    \centering
    \begin{subfigure}[b]{0.3\textwidth}
    \begin{tikzpicture}
    \useasboundingbox (-3,0) rectangle (3,2);
	\node[](A) at (-2, 0) {Alice};
	\node[](B) at (2, 0) {Bob};
	\node[](R) at (0, 2) {Referee};
	
	\draw[->] (A) [out=45, in=225] to (R);
	\draw[->] (B) [out=135, in=315] to (R);
\end{tikzpicture}
\caption{SMP}
\end{subfigure}
\hfill
\begin{subfigure}[b]{0.3\textwidth}
    \begin{tikzpicture}
    \useasboundingbox (-3,0) rectangle (3,2);
	\node[](A) at (-2, 1) {Alice};
	\node[](B) at (2, 1) {Bob};
	
	\draw[->] (A) [out=0, in=180] to (B);
\end{tikzpicture}
\caption{One-way}
\end{subfigure}
\hfill
\begin{subfigure}[b]{0.3\textwidth}
    \begin{tikzpicture}
    \useasboundingbox (-3,0) rectangle (3,2);
	\node[](A) at (-2, 1) {Alice};
	\node[](B) at (2, 1) {Bob};
	
	\draw[transform canvas={yshift=2mm}][->] (A) to (B);
	\draw[transform canvas={yshift=0mm}][<-] (A) to (B);
	\draw[transform canvas={yshift=-6mm}][->](A) to node[above,midway] {$\vdots$} (B) ;
\end{tikzpicture}
\caption{Interactive}
\end{subfigure}
    \caption{We consider 3 different communication settings. (a) Simultaneous message passing (SMP), where Alice and Bob each sends a message to the Referee, who computes a function of the two messages. (b) One-way communication, where Alice sends a message to Bob, and Bob performs a measurement on his states conditioned on Alice's message. (c) Arbitrary rounds of local quantum operations and classical communication (LOCC), where the operation in each round is conditioned on all messages in previous rounds.}
    \label{fig:communication}
\end{figure}

Our algorithm with single-copy measurements is a variant of the one proposed in \cite{Elben2020cross}, where they also observed similar performance via numerical simulation. 

There are multiple lenses to view our result:

\begin{itemize}
    \item 
The first is as a practical answer to the aforementioned question of comparing (or calibrating) the outputs of two quantum systems that don't share a quantum channel. From this point of view, with the more experimentally relevant single-copy measurements,  the sample cost of $O(\sqrt{d})$ (when $\varepsilon$ is a constant) represents a saving from the naive approach of doing full tomography on each system's output state, which takes $\Omega(d^3)$ copies~\cite{Haah2017sample}, making the task experimentally feasible with $\sim 20-30$ qubits instead of just a few qubits using full tomography. On the other hand, the task is not scalable in general as  our lower bound shows that the task still requires an exponential number of copies (relative to the number of qubits) even with arbitrary multi-copy and adaptive measurements.

\item The second is to view the result in contrast to similar problems that are sample-efficient. 
Specifically, if instead of both Alice and Bob having unknown quantum states, the goal is to compare Alice's unknown state to the classical description of a known state, then this task reduces to direct fidelity estimation~\cite{flammia2011direct,dasilva2011practical} and can be done using  $O(1/\varepsilon^2)$ copies, via for example the classical shadows algorithm~\cite{Huang2020predicting}.
In fact, as our task only involves learning a specific property of Alice and Bob's quantum systems, one might expect that it is analogous to \emph{shadow tomography}, which is known to be sample-efficient in many cases~\cite{Aaronson2018shadow,Aaronson2019gentle,Huang2020predicting,Badescu2021improved}. In contrast, our result shows that learning the properties of multiple quantum systems in a distributed setting is very different from learning a single system, and the classical communication constraint is a significant obstacle for sample-efficiency.

\item The third lens involves focusing on the result's sample complexity staying the same under different measurement and communication models. This is surprising because these different models are known to have very different power for other tasks.  With respect to the power of multi-copy measurements, the sample complexity of learning an unknown mixed state $\rho$ acting on  $\mathbb{C}^{d}$ is $\Omega(d^3)$ with single-copy measurements, and improves to  $O(d^2)$ with multi-copy measurements~\cite{ODonnell2016efficient,ODonnell2017efficient,Haah2017sample}. Moreover, exponential separations in sample complexity between single and multi-copy measurements have been shown in several single system quantum property testing tasks~\cite{Bubeck2020entanglement,aharonov2021quantum,Chen2021exponential}. 
In the distributed setting, our results provide a counterexample to this mounting evidence that coherent access to multiple copies provides an advantage in learning properties of quantum systems.
In contrast with our results that do not vary with the different communication models described in Fig.~\ref{fig:communication}, we remark that there are certain partial functions that achieve exponential separations in communication complexity \cite{KNTZ01}.

\item Finally, we can step back and consider the model of distributed quantum property testing, where the goal is to estimate some linear function $f(\rho, \sigma)$ of the states with Alice and Bob, where there is no quantum channel between them -- a setting motivated by hardware constraints as it is hard to perform cross-platform operations. With this view, our paper resolves the sample complexity for estimating $f(\rho, \sigma)= \Tr(\rho \sigma)$.   It is an intriguing open question to investigate whether the features of our results -- both the sample complexity and the fact that the complexity doesn't change for a variety of classical communication protocols and access to multi-copy measurements -- come from the special structure of inner product, or are general and apply to other distributed quantum estimation tasks. 
\end{itemize}

Recently a comprehensive framework was established that considered different models of how an experimentalist can estimate properties of a physical system~\cite{aharonov2021quantum}. Looking into the future we will have many different physical platforms for quantum computing, as well as limited quantum communication networks. Therefore it is well motivated to generalize this framework to the distributed setting and study the power of different models, where besides the classical communication considered in this work we can also allow limited amounts of quantum communication across different systems. These efforts will be relevant in both the near-term era for quantum device characterization~\cite{Eisert2020,Kliesch2021theory,Carrasco2021Theoretical,Greganti2021cross}, and the fault-tolerant era for distributed quantum information processing~\cite{Wehner2018quantum}.

\section{Overview}
Our task, abbreviated as $\ipe$, is formally defined as follows.

\begin{definition}[Inner product estimation, $\ipe$]
Alice is given $k$ copies of an unknown state $\rho$ acting on $\mathbb{C}^{d}$ and Bob is given $k$ copies of an unknown state $\sigma$ acting on $\mathbb{C}^{d}$. Their goal is to estimate $\Tr(\rho\sigma)$ up to additive error $\varepsilon\in(0,1)$ with success probability at least $2/3$, using local quantum operations and classical communication (LOCC).
\end{definition}

We will often work with the following decision problem. 
\begin{definition}[Decisional inner product estimation, DIPE]
\label{def:DIPE}
Alice and Bob are each given $k$ copies of a pure state in $\mathbb{C}^d$. They are promised that one of the following two cases hold:
\begin{enumerate}
    \item Alice and Bob both have $\ket{\phi}^{\otimes k}$, for a uniformly random state $\ket{\phi}\sim\mathbb{C}^d$.
    \item Alice has $\ket{\phi}^{\otimes k}$ and Bob has $\ket{\psi}^{\otimes k}$, where $\ket{\phi}$ and $\ket{\psi}$ are independent uniformly random states in $\mathbb{C}^d$.
\end{enumerate}
Their goal is to decide which case they are in with success probability at least $2/3$, using an interactive protocol that involves local quantum operations and classical communication.
\end{definition}

DIPE is a special instance of $\ipe$. To see this, notice that in case 1, the inner product between Alice and Bob's state is 1, while in case 2 the inner product is
\begin{equation}
    \E_{\phi,\psi\sim\mathbb{C}^d}\Tr(\phi\psi)=\E_{\phi,\psi\sim\mathbb{C}^d}\left|\braket{\phi}{\psi}\right|^2=\frac{1}{d}
\end{equation}
on average, which is exponentially small. Here we use $\phi$ to represent the density matrix $\ketbra{\phi}$. Therefore, if Alice and Bob can solve $\ipe$ with accuracy say $\varepsilon=0.1$, then they can also solve DIPE with high probability. Our main result thus implies that DIPE can be solved with single-copy measurements using $k=O(\sqrt{d})$ copies.

The reason for introducing DIPE is that it is the main problem for showing our lower bounds. We prove a $\Omega(\sqrt{d})$ lower bound for DIPE against arbitrary interactive protocols with multi-copy measurements, even though the input states in DIPE are only limited to pure states. Based on this lower bound, we then show a $\Omega(\sqrt{d}/\varepsilon)$ lower bound for estimating the inner product between pure states via a reduction to DIPE, which implies the lower bound in our main result.

\subsection{Algorithms in the symmetric subspace}
In the following we show how to develop algorithms with multi-copy measurements to solve DIPE, and develop tools towards proving lower bounds. 

Our starting point is to notice that locally, Alice's state has the same distribution between the two cases, which is $\ket{\psi}^{\otimes k}$ for a uniformly random state $\ket{\psi}\sim\mathbb{C}^d$, and the same for Bob. This implies that some kind of communication is necessary, as Alice and Bob cannot distinguish between the two cases by themselves. Furthermore, notice that states of the form $\ket{\psi}^{\otimes k}$ have a special structure, due to the invariance under permutation of different subsystems. We will heavily exploit this structure in both algorithms and lower bounds, which is captured by the symmetric subspace.

\begin{definition}[Symmetric subspace]
\label{def:symmetricsubspace}
The symmetric subspace of $(\mathbb{C}^d)^{\otimes k}$ is defined by
\begin{equation}
    \ssym{d}{k}=\left\{\ket{\omega}\in(\mathbb{C}^d)^{\otimes k} : P_d(\pi)\ket\omega = \ket\omega,\forall \pi\in S_k\right\}.
\end{equation}
Here $P_d(\pi)$ is a unitary operator given by $P_d(\pi)= \sum_{i_1,\ldots,i_k=0}^{d-1} \ket{i_{\pi^{-1}(1)}, \ldots, i_{\pi^{-1}(k)}}\bra{i_1,\ldots,i_k}$, and $\pi$ is a permutation on $k$ items. An equivalent definition of the symmetric subspace is
\begin{equation}
    \ssym{d}{k}=\vspan\left\{\ket{\psi}^{\otimes k}: \ket{\psi}\in\mathbb{C}^d\right\}.
\end{equation}
Here $\vspan$ denotes linear combination with complex coefficients. Let $T$ be the type function that acts on vectors $i=(i_1,\dots,i_k)\in\{0,\dots,d-1\}^k$ as $T(i)=(\ell_0,\dots,\ell_{d-1})$, where $\ell_j$ is the number of times that $j$ appears in $(i_1,\dots,i_k)$. Another equivalent definition of the symmetric subspace is
\begin{equation}
    \ssym{d}{k}=\vspan\left\{\sum_{i:T(i)=\ell}\ket{i_1,\dots,i_k}:\ell=(\ell_0,\dots,\ell_{d-1}),\ell_j\geq 0,\sum_{j=0}^{d-1}\ell_j=k\right\}.
\end{equation}
\end{definition}

The equivalence of the three definitions is non-trivial. We refer to \cite{harrow2013church} for the proof of this fact as well as other facts about the symmetric subspace that we use.

From the definition, it is clear that no matter which case they are in, Alice and Bob's state is always inside the symmetric subspace. Therefore they only need to construct measurement operators that are supported on $\ssym{d}{k}$ as well. To characterize such measurements, it is helpful to consider the orthogonal projector on to $\ssym{d}{k}$.

\begin{lemma}\label{lemma:symprojector}
Let $\psym{d}{k}$ be the orthogonal projector on to $\ssym{d}{k}$, that is, the positive Hermitian operator with eigenspace $\ssym{d}{k}$ which satisfies $\psym{d}{k}\cdot\psym{d}{k} =\psym{d}{k}$. This operator has the following equivalent forms
\begin{equation}
    \psym{d}{k}=\binom{d+k-1}{k}\E_{\phi\sim\mathbb{C}^d}\ketbra{\phi}^{\otimes k}=\frac{1}{k!}\sum_{\pi\in S_k}P_d(\pi).
\end{equation}
In particular, the dimension of the symmetric subspace is
\begin{equation}
    \mathrm{dim}\ssym{d}{k}=\Tr(\psym{d}{k})=\binom{d+k-1}{k}.
\end{equation}
\end{lemma}

Suppose Alice measures her state $\ket{\phi}^{\otimes k}$ with a POVM $\{M_i\}$. The probability of seeing the outcome $i$ is given by
\begin{equation}
    \Tr(M_i\ketbra{\phi}^{\otimes k})=\Tr(M_i\psym{d}{k}\ketbra{\phi}^{\otimes k}\psym{d}{k})=\Tr(\psym{d}{k} M_i\psym{d}{k}\ketbra{\phi}^{\otimes k}).
\end{equation}
Therefore, we can assume without loss of generality that the POVM is supported on $\ssym{d}{k}$, that is, it satisfies $M_i=\psym{d}{k} M_i\psym{d}{k}$ and $\sum_i M_i=\psym{d}{k}$. We single out the following POVM as the ``standard POVM" in $\ssym{d}{k}$.

\begin{definition}[Standard POVM]
It is defined as the following continuous POVM:
\begin{equation}\label{eq:standardpovm}
    \left\{\binom{d+k-1}{k}\ketbra{u}^{\otimes k}\dd u\right\}.
\end{equation}
Here $\dd u$ denotes the uniform measure over pure states in $\mathbb{C}^d$.
\end{definition}

\begin{figure}[t]
\centering
\begin{subfigure}[b]{0.48\textwidth}
\fbox{\parbox[t][5.5cm][s]{\textwidth}{
Alice and Bob each has $k$ copies of an unknown state.        
\begin{enumerate}
            \item Alice measures all copies with the standard POVM, gets result $u$.
            \item Bob measures all copies with the standard POVM, gets result $v$.
            \item They compute a function of $u,v$.
        \end{enumerate}
}}
\caption{Multi-copy}
\end{subfigure}
\hfill
\begin{subfigure}[b]{0.48\textwidth}
\fbox{\parbox[t][5.5cm][s]{\textwidth}{
    Alice and Bob each has $k$ copies of an unknown state. They share a random unitary matrix $U$.
        \begin{enumerate}
            \item Alice rotates each copy of her state by $U$ and measures in the computational basis, gets $A=\{a_1,\dots,a_k\}$.
            \item Bob rotates each copy of his state by $U$ and measures in the computational basis, gets $B=\{b_1,\dots,b_k\}$.
            \item They compute a function of $U,A,B$.
        \end{enumerate}
}}
\caption{Single-copy}
\end{subfigure}
\caption{Algorithm templates for distributed quantum estimation tasks. Both algorithms only use simultaneous message passing for communication (Fig.~\ref{fig:communication}a). The multi-copy algorithm (a) is a deterministic strategy, while the single-copy algorithm (b) is randomized and requires shared randomness between Alice and Bob.}
\label{fig:algtemplate}
\end{figure}

Lemma~\ref{lemma:symprojector} ensures that Eq.~\eqref{eq:standardpovm} is indeed a POVM in $\ssym{d}{k}$. What is special about it and why do we call it ``standard"? As we now show, the motivation for considering this POVM in our tasks comes from the fact that it is the optimal POVM for performing tomography on an unknown pure state. More precisely, given $\ket{\phi}^{\otimes k}$ as $k$ copies of an unknown pure state, consider the following simple algorithm: Measure $\ket{\phi}^{\otimes k}$ with the standard POVM, and output the measurement outcome $\ket{u}$. The average fidelity is given by
\begin{equation}\label{eq:tomography}
    \begin{split}
        \E\left|\braket{\phi}{u}\right|^2&=\E_{u\sim\mathbb{C}^d}\left|\braket{\phi}{u}\right|^2\cdot\binom{d+k-1}{k}\left|\braket{\phi}{u}\right|^{2k}\\
        &=\binom{d+k-1}{k}\E_{u\sim\mathbb{C}^d}\Tr(\phi^{\otimes k+1}\cdot u^{\otimes k+1})\\
        &=\binom{d+k-1}{k}\binom{d+k}{k+1}^{-1}\Tr(\psym{d}{k+1} \phi^{\otimes k+1})\\
        &=\frac{k+1}{d+k}.
    \end{split}
\end{equation}
Here the third line follows from Lemma~\ref{lemma:symprojector}. It is shown in~\cite{harrow2013church} that the above average fidelity is the highest when $\phi$ is given uniformly at random. This implies that $k=O(d)$ copies is necessary and sufficient to learn an unknown pure state with 0.99 fidelity using multi-copy measurements.

\begin{lemma}[See e.g. \cite{harrow2013church}]
Let $\ket{\phi}^{\otimes k}$ be $k$ copies of an unknown pure state, where $\ket{\phi}\sim\mathbb{C}^d$ is uniformly random. Suppose we would like to do full tomography: learn $\ket{\phi}$ with as few copies as possible. Consider the general strategy of measuring $\ket{\phi}^{\otimes k}$ with a POVM and preparing a state $\ket{\phi'}$ based on the measurement outcome. Then the standard POVM is the optimal measurement for this task, in the sense that it achieves the optimal average fidelity
\begin{equation}
    \E\left|\braket{\phi}{\phi'}\right|^2
\end{equation}
where the expectation is over the randomness in $\ket{\phi}$ as well as the randomness in quantum measurement.
\end{lemma}

Given the optimality of the standard POVM in tomography, perhaps it is useful for inner product estimation as well? Here the main challenges are two-fold. First, we have two unknown quantum systems instead of one, so we need to design a distributed and possibly interactive algorithm. Second, we want to use much less than $O(d)$ copies, as otherwise we can just perform tomography.

Here we propose an algorithm template for solving the distributed inner product estimation problem, which can also be applied to other distributed estimation problems. The idea (described in Fig.~\ref{fig:algtemplate}a) is that Alice and Bob each performs the standard POVM, and then they compute a function of their measurement outcomes. Note that this algorithm can also be implemented in the SMP communication model, as Alice and Bob can send their results to a Referee, who then computes the function. Despite its simple and non-adaptive nature, we show that $k=O(\sqrt{d})$ copies suffices to solve DIPE in this model.
\begin{theorem}
There exists an algorithm (using the template in Fig.~\ref{fig:algtemplate}a) that solves DIPE using $k=O(\sqrt{d})$ copies, with success probability at least $2/3$.
\end{theorem}

In this algorithm, after receiving Alice and Bob's measurement results $u$ and $v$, the Referee simply computes $\left|\braket{u}{v}\right|^2$, and then decides case 1 or case 2 based on its value. To see how this works, note that Eq.~\eqref{eq:tomography} implies that the overlap between Alice's state $\ket{\phi}$ and her measurement outcome $u$ is roughly $\left|\braket{\phi}{u}\right|\approx \sqrt{k/d}$. Using this observation, in case 1 where Alice and Bob have the same state, we show that their measurement results roughly have overlap
\begin{equation}
    \left|\braket{u}{v}\right|\approx\frac{k}{d}.
\end{equation}
On the other hand, in case 2 where Alice and Bob have independent random states, their measurement results are uncorrelated, and we show that $\left|\braket{u}{v}\right|\approx\frac{1}{\sqrt{d}}$. Therefore $k=O(\sqrt{d})$ is sufficient to distinguish between the two cases. See Section~\ref{sec:multicopy} for a detailed analysis.

Next, we show that this algorithm can be generalized to solve $\ipe$ for pure states.

\begin{theorem}
Suppose Alice has $\ket{\phi}^{\otimes k}$ and Bob has $\ket{\psi}^{\otimes k}$, for arbitrary unknown pure states $\ket{\phi},\ket{\psi}\in\mathbb{C}^d$. There exists an algorithm (using the template in Fig.~\ref{fig:algtemplate}a) that estimates $\Tr(\phi\psi)=\left|\braket{\phi}{\psi}\right|^2$ within additive error $\varepsilon\in(0,1)$ using $k=O(\max\{1/\varepsilon^2,\sqrt{d}/\varepsilon\})$ copies, with success probability at least $2/3$.
\end{theorem}

Here as well, the Referee computes a simple function of $\left|\braket{u}{v}\right|^2$ as the estimator for $\left|\braket{\phi}{\psi}\right|^2$. The analysis of this algorithm involves heavy calculations and is presented in Section~\ref{sec:multicopy}.

Note that the above algorithms require that the input states are i.i.d. copies of a pure state. When the input is $\rho^{\otimes k}$ for some mixed state $\rho$, the above algorithms do not apply as $\rho^{\otimes k}$ is in general not supported on $\ssym{d}{k}$. We do not consider multi-copy measurements on mixed states in this paper because single-copy measurements already suffice (Theorem~\ref{thm:overviewsinglecopy}). However, we believe that similar sample complexity can be achieved if we replace the standard POVM (which only works for pure states) with the POVM used in optimal mixed state tomography~\cite{ODonnell2016efficient,ODonnell2017efficient,Haah2017sample}.

\subsection{Lower bounds against multi-copy measurements}

Following the algorithms presented above which only uses SMP communication, a natural question is can Alice and Bob solve the inner product estimation problem using fewer copies with an interactive algorithm? Surprisingly, we show that interaction and adaptive strategies cannot improve the sample complexity by more than a constant factor.

We start by analyzing the following simple algorithm for DIPE which uses one-way communication (Fig.~\ref{fig:communication}b):
\begin{enumerate}
    \item Alice measures all copies of her state with the standard POVM, gets result $u$, and sends to Bob.
    \item Conditioned on $u$, Bob performs a two-outcome measurement $\{M_u,\psym{d}{k}-M_u\}$ on all copies of his state, and decides case 2 if he sees $M_u$, and case 1 otherwise.
\end{enumerate}

This algorithm is a natural generalization of the algorithm template shown in Fig.~\ref{fig:algtemplate}a. Here the main difference is adaptiveness: after seeing Alice's measurement outcome $u$, Bob can choose his measurement based on $u$, and therefore for different $u$ he can perform different measurements. This model is more powerful than the algorithm in Fig.~\ref{fig:algtemplate}a, as Bob can always perform the standard measurement and simulate the SMP communication by himself. Therefore the main question is what is Bob's optimal measurement to distinguish between the two cases, and what performance it achieves.

In the following we show how to exactly identify Bob's optimal measurement strategy, and then show that it still requires $\Omega(\sqrt{d})$ copies to solve DIPE. To start with, we consider the distinguishing task that Bob needs to solve after receiving Alice's message. In case 2 when they have independent random states, Bob's state is always the maximally mixed state (in $\ssym{d}{k}$) 
\begin{equation}
    \mmix:=\frac{1}{\binom{d+k-1}{k}}\psym{d}{k}.
\end{equation}
In case 1 when they have same random states, Bob's state is ``updated" after seeing the outcome $u$ as
\begin{equation}
    \rho_u:=\binom{d+k-1}{k}\E_{\phi\sim\mathbb{C}^d}\ketbra{\phi}^{\otimes k}\left|\braket{u}{\phi}\right|^{2k}.
\end{equation}
One way to see this is by Bayes' theorem $\Pr[\phi|u]=\Pr[u|\phi]=\binom{d+k-1}{k}\left|\braket{u}{\phi}\right|^{2k}$ as both $u$ and $\phi$ are marginally uniformly distributed. Then, the optimal distinguishing measurement is related to the trace distance,
\begin{equation}\label{eq:tracedistance}
    \max_{0\leq M_u\leq \psym{d}{k}}\Tr(M_u(\mmix-\rho_u))=\frac{1}{2}\|\rho_u-\mmix\|_1.
\end{equation}

What is the optimal measurement to distinguish $\mmix$ and $\rho_u$? Intuitively, Bob's knowledge is that if his state is in case 1, then from the expression $\Pr[\phi|u]\propto\left|\braket{u}{\phi}\right|^{2k}$ it will be slightly biased towards the vector $\ket{u}$, while in case 2 his state is uniform. Therefore, a projector onto a subspace without $\ket{u}$ would have larger overlap with $\mmix$ and smaller overlap with $\rho_u$. Motivated by this observation, we consider the following subspace of $\ssym{d}{k}$,
\begin{equation}
    W_u:=\mathrm{span}\left\{\ket{v}^{\otimes k}:\ket{v}\in\mathbb{C}^d, \braket{v}{u}=0\right\},
\end{equation}
and let $\Pi_{W_u}$ be the orthogonal projector onto $W_u$. We show via a sequence of arguments that $\Pi_{W_u}$ indeed achieves the maximum in LHS of Eq.~\eqref{eq:tracedistance}. Furthermore, using this we can prove that $\frac{1}{2}\|\rho_u-\mmix\|_1=O(k^2/d)$, which leads to the following lower bound.

\begin{theorem}\label{thm:overviewoneway}
Consider a one-way protocol where Alice performs the standard measurement and sends the result to Bob. Then $k=\Omega(\sqrt{d})$ copies are necessary for Alice and Bob to solve DIPE.
\end{theorem}

There are two important issues that we need to solve in order to obtain a lower bound against arbitrary interactive protocols. First, in the above argument we only analyzed Bob's optimal strategy if Alice performs the standard measurement. However this is not necessarily the overall optimal strategy, as it is possible that Alice could perform some other measurement so that Bob could do much better. Second, Alice and Bob could potentially use a multi-round interactive protocol to beat the $\Omega(\sqrt{d})$ lower bound.

In the following we sketch a proof of a $\Omega(\sqrt{d})$ lower bound for DIPE against arbitrary interactive protocols. First consider the general one-way protocol: Alice performs a POVM $\{M_i\}$ on her copies and sends the result $i$ to Bob. After receiving $i$, Bob performs a two-outcome POVM on his copies to decide which case they are in. Same as before, in case 2 Bob's state always equals the maximally mixed state $\sigma_{\mathrm{m}}$. On the other hand, in case 1 Bob's state gets updated after seeing $i$, given by
\begin{equation}
\label{eq:Mi}
    \rho_i=\frac{\binom{d+k-1}{k}}{\Tr(M_i\psym{d}{k})}\cdot\E_{\phi\sim\mathbb{C}^d}\Tr(M_i \ketbra{\phi}^{\otimes k})\ketbra{\phi}^{\otimes k}.
\end{equation} 
We show that when $k=o(\sqrt{d})$, $\rho_i$ is indistinguishable from $\sigma_{\mathrm{m}}$ no matter what $M_i$ is, which implies a $\Omega(\sqrt{d})$ lower bound for DIPE against one-way protocols, and then generalize this argument against arbitrary interactive protocols.

The indistinguishability between $\rho_i$ and $\sigma_{\mathrm{m}}$ follows from a connection to optimal measure-and-prepare channels. First, we change our view on $M_i$, from a measurement to a state (appropriately normalized), and then consider it as an input to the quantum channel
\begin{equation}
    \mathrm{MP}(\tau):=\binom{d+k-1}{k}\E_{\phi\sim\mathbb{C}^d}\Tr(\tau\cdot \ketbra{\phi}^{\otimes k})\ketbra{\phi}^{\otimes k}.
\end{equation}
From Eq.~\eqref{eq:Mi}, Bob's state $\rho_i$ is the output state of this channel on the input $\frac{\psym{d}{k}M_i\psym{d}{k}}{\Tr(M_i\psym{d}{k})}$. The channel $\mathrm{MP}$ is a well known measure-and-prepare channel, which was related to optimal cloning by Chiribella~\cite{Chiribella2011on}. Using this connection, we argue that when $k=o(\sqrt{d})$, the output of $\mathrm{MP}$ on arbitrary input is indistinguishable from $\sigma_{\mathrm{m}}$, which is shown via an upper bound of the max-relative entropy between $\rho_i$ and $\mmix$.

\begin{definition}[Max-relative entropy]
The max-relative entropy of two mixed states $\rho$ and $\sigma$ is defined as
\begin{equation}
    D_{\max}(\rho\|\sigma):=\inf\left\{\lambda\in\mathbb{R}:\rho\leq e^{\lambda}\sigma\right\}.
\end{equation}
Here $A\geq B$ means that $A-B$ is positive semi-definite.
\end{definition}

Closeness in max-relative entropy can be understood as a form of multiplicative closeness, which is very strong as it implies closeness in many other distances including the trace distance. We show that $D_{\max}(\mmix\|\rho_i)\leq k^2/d$ always holds regardless of Alice's strategy. This implies a stronger version of Theorem~\ref{thm:overviewoneway} that holds against arbitrary one-way protocols.

Interestingly, we then show that the above closeness in max-relative entropy is sufficient to prove lower bound against arbitrary interactive protocols. The argument does not follow if we only have, for example, closeness in trace distance.

\begin{theorem}
$k=\Omega(\sqrt{d})$ copies is necessary for Alice and Bob to solve DIPE, even when they are allowed arbitrary interactive protocols (or arbitrary LOCC operations).
\end{theorem}

Finally, we construct a lower bound instance for estimating the inner product of arbitrary pure states with accuracy $\varepsilon$, which allows a reduction to DIPE. This proves the following general lower bound.

\begin{theorem}\label{thm:overviewlowerbound}
Suppose Alice has input $\ket{\phi}^{\otimes k}$ and Bob has input $\ket{\psi}^{\otimes k}$, for arbitrary unknown pure states $\ket{\phi},\ket{\psi}\in\mathbb{C}^{d}$. Then $k=\Omega(\sqrt{d}/\varepsilon)$ copies is necessary for them to estimate $\Tr(\phi\psi)$ up to additive error $\varepsilon\in(0,1)$ with success probability $2/3$, even when they are allowed arbitrary interactive protocols (or arbitrary LOCC operations).
\end{theorem}

Recall that in the non-distributed setting, $k=O(1/\varepsilon^2)$ copies suffice to solve $\ipe$ using SWAP test. In Section~\ref{sec:swaptest} we discuss the SWAP test in more detail, and also show it is optimal by proving a $\Omega(1/\varepsilon^2)$ lower bound, even when the input states are pure. Clearly, this lower bound proven in the non-distributed setting also applies to our distributed setting. Combining the two lower bounds gives the lower bound in our main result.

\subsection{Single-copy measurement suffices}
We now turn to the weaker and more experimentally relevant setting of single-copy measurements. Here, Alice and Bob can only measure their states one copy at a time. At any given time step, when Alice is measuring a copy of her state, her choice of the measurement basis could depend on her previous measurement outcomes as well as the messages she received from Bob.

In the following we consider the algorithm template in Fig.~\ref{fig:algtemplate}b, which is in the weakest measurement and communication setting. Here Alice and Bob only communicate in the SMP model, and are allowed to have shared randomness so that they can implement the same random unitary matrix $U$. Moreover, their measurement strategy is non-adaptive, as they always measure in the same basis, independent of previous outcomes. We show that an algorithm in this setting already suffices to essentially achieve the  optimal sample complexity.

\begin{theorem}\label{thm:overviewsinglecopy}
Suppose Alice has $\rho^{\otimes k}$ and Bob has $\sigma^{\otimes k}$, for arbitrary unknown mixed states $\rho,\sigma$ acting on $\mathbb{C}^{d}$. There exists an algorithm (using the template in Fig.~\ref{fig:algtemplate}b) that estimates $\Tr(\rho\sigma)$ within additive error $\varepsilon\in(0,1)$ using $k=O(\max\{1/\varepsilon^2,\sqrt{d}/\varepsilon\})$ copies, with success probability at least $2/3$.
\end{theorem}

A variant of our algorithm was analyzed numerically in \cite{Elben2020cross} which observed similar performance. It is interesting to point out that the function computed by the algorithm in the 3rd step (Fig.~\ref{fig:algtemplate}b) actually does not depend on $U$ -- the frequencies of the outcome bit strings are all we need. This is because of the rotation invariance of the inner product $\Tr(\rho\sigma)=\Tr(U\rho U^\dag\cdot U\sigma U^\dag)$.

The main idea of the algorithm, as well as the main spirit of the template in Fig.~\ref{fig:algtemplate}b, follows from the idea that quantum property estimation problems can often be reduced to classical property estimation problems after we apply a random projection on the unknown quantum state, which can be implemented by a random unitary rotation with computational basis measurement. This idea has led to powerful algorithms that achieve optimal performance with single-copy measurements~\cite{Huang2020predicting,Bubeck2020entanglement}. In our task, we show that the quantum inner product estimation problem reduces to a classical inner product estimation problem of two discrete distributions
\begin{equation}
    p_b:=\expval{U\rho U^\dag}{b},\,\,\,\,q_b:=\expval{U\sigma U^\dag}{b},
\end{equation}
where $b\in\{0,\dots,d-1\}$. We use a standard collision estimator that counts the number of pairwise collisions among the measurement outcomes to estimate this classical inner product $\sum_b p_b q_b$, and show that $k=O(\max\{1/\varepsilon^2,\sqrt{d}/\varepsilon\})$ measurement samples suffice.

Although this algorithm achieves similar performance as the algorithm with multi-copy measurements, they can be very different in some aspects. First, the single-copy algorithm (Fig.~\ref{fig:algtemplate}b) requires shared randomness between Alice and Bob, while the multi-copy algorithm (Fig.~\ref{fig:algtemplate}a) is a deterministic strategy. Second, the multi-copy algorithm can be more efficient if we have some prior knowledge such as an upper bound on the inner product (see Theorem~\ref{thm:multicopyalg} and Theorem~\ref{thm:singlecopyalg}). 

In addition, a natural question is can Alice and Bob still achieve the optimal sample complexity using single-copy measurements and simultaneous message passing, but without shared randomness. In Appendix~\ref{app:independentshadow} we analyze a natural algorithm in this setting where Alice and Bob perform independent classical shadow estimation, and show that this algorithm requires $O(d)$ samples. This gives evidence that shared randomness is necessary to achieve the optimal sample complexity (using single-copy measurements and SMP), and we leave the proof of this statement for future work.

\section{Discussion}
\label{sec:discussion}

Our results for distributed quantum inner product estimation can be applied to prove upper and lower bounds for estimating several other quantities. For example, our algorithm with single-copy measurements (Theorem~\ref{thm:overviewsinglecopy}) implies $O(\max\{1/\varepsilon^2,\sqrt{d}/\varepsilon\})$ sample complexity for estimating purity $\Tr(\rho^2)$ within additive error $\varepsilon$. Combining purity estimation with inner product estimation, we can estimate fidelity-like measures by renormalizing $\Tr(\rho\sigma)$ with $\Tr(\rho^2)$ and $\Tr(\sigma^2)$ (see \cite{Liang2019quantum,Elben2020cross} for more details). We can also estimate the squared Hilbert-Schmidt distance
\begin{equation}
    D_{\mathrm{HS}}^2(\rho,\sigma):=\Tr((\rho-\sigma)^2)=\Tr(\rho^2)+\Tr(\sigma^2)-2\Tr(\rho\sigma)
\end{equation}
in the distributed setting with single-copy measurements which has the same sample complexity. In fact our lower bound for inner product estimation also implies a tight lower bound for estimating $D_{\mathrm{HS}}^2(\rho,\sigma)$.

\begin{corollary}
There exists an algorithm for estimating $D_{\mathrm{HS}}^2(\rho,\sigma)$ within additive error $\varepsilon\in(0,1)$ with non-adaptive single-copy measurements and simultaneous message passing that uses $k=O(\max\{1/\varepsilon^2,\sqrt{d}/\varepsilon\})$ copies. Furthermore, any algorithm for estimating $D_{\mathrm{HS}}^2(\rho,\sigma)$ with arbitrary rounds of adaptive multi-copy measurements and classical communication must use at least \\ $\Omega(\max\{1/\varepsilon^2,\sqrt{d}/\varepsilon\})$ copies.
\end{corollary}

To see the lower bound, note that our lower bound for inner product estimation holds even for pure states (Theorem~\ref{thm:overviewlowerbound}). For pure states $\rho$ and $\sigma$, the inner product and $D_{\mathrm{HS}}^2$ is related by
\begin{equation}
    \Tr(\rho\sigma)=1-\frac{D_{\mathrm{HS}}^2(\rho,\sigma)}{2},
\end{equation}
which directly implies a $\Omega(\max\{1/\varepsilon^2,\sqrt{d}/\varepsilon\})$ lower bound for estimating $D_{\mathrm{HS}}^2$ with pure input states.

Finally, it is an interesting future direction to apply our algorithm templates in Fig.~\ref{fig:algtemplate} as well as our techniques for proving lower bounds to other distributed quantum property estimation problems. For example, we can consider estimating linear functions $f(\rho,\sigma)=\Tr(W\cdot \rho\otimes \sigma)$ in the distributed setting, where $W$ is an arbitrary Hermitian operator supported on $\mathbb{C}^d\otimes \mathbb{C}^d$. Our results settle the sample complexity for estimating $W=\swapop$. Clearly the sample complexity should depend on some property of $W$, as the task is easy when $W=W_A\otimes W_B$ is separable. In addition, it is also interesting to consider estimating multiple linear functions at the same time using ideas from shadow tomography~\cite{Aaronson2018shadow,Aaronson2019gentle,Huang2020predicting,Badescu2021improved}.

\section{Inner product estimation with multi-copy measurements}
\label{sec:multicopy}

In this section we develop algorithms for solving DIPE as well as $\ipe$ with multi-copy measurements. These algorithms only use the simultaneous message passing model for communication (Fig.~\ref{fig:communication}) and are instances of the template in Fig.~\ref{fig:algtemplate}a.

We start with the following technical lemma which is used frequently in the paper and can be understood as follows: to sample a random $d$-dimensional unit vector, we can first sample a random $d-1$-dimensional unit vector, and then sample the $d$-th entry from an appropriate distribution, and then renormalize the vector.

\begin{lemma}[Decomposition of a random state]\label{lemma:decomposition}
Let $\ket{\phi}$ be a random unit vector in $\mathbb{C}^d$, and let $\{\ket{0},\ket{1},\dots,\ket{d-1}\}$ be an arbitrary orthonormal basis of $\mathbb{C}^d$. Then $\ket{\phi}$ can be decomposed into 3 independent random variables $\alpha,\theta,\ket{v}$ as
\begin{equation}
    \ket{\phi}=\alpha e^{i\theta}\ket{0}+\sqrt{1-\alpha^2}\ket{v}.
\end{equation}
Here,
\begin{itemize}
    \item $\alpha\in[0,1]$ is a positive real random variable with density $2(d-1)(1-x^2)^{d-2}x$, and $\alpha^2$ has density $(d-1)(1-x)^{d-2}$, where $x\in[0,1]$,
    \item $\theta\in[0,2\pi]$ is a uniformly random phase,
    \item $\ket{v}$ is a random unit vector in $\vspan\{\ket{1},\dots,\ket{d-1}\}$.
\end{itemize}
\end{lemma}
\begin{proof}
The proof follows from basic properties of the distribution of random unit vector in $\mathbb{C}^d$ and we refer to \cite[Chapter 4.2]{Hiai2006semicircle} and \cite{Nechita2007} for more details.

A random unit vector in $\mathbb{C}^d$ can be characterized as\footnote{Technically the distribution of a random pure state should be the distribution of a random unit vector modulo the global phase. Here for simplicity we just think of random pure states as random unit vectors; this will not affect any results.}
\begin{equation}
    \ket{\phi}=\frac{\left(\xi_1,\dots,\xi_d\right)^T}{\sqrt{|\xi_1|^2+\cdots +|\xi_d|^2}}
\end{equation}
where $\xi_j$ are independent standard complex normal random variables. $\xi_j$ has the following properties: it can be written as $\xi_j=|\xi_j|e^{i\theta_j}$, where $\theta_j\in[0,2\pi]$ is a uniformly random phase and $|\xi_j|^2$ has exponential distribution with density $e^{-x}$ ($x>0$). Then we have
\begin{equation}
\begin{split}
    \ket{\phi}&=\frac{\left(\xi_1,0,\dots,0\right)^T}{\sqrt{|\xi_1|^2+\cdots +|\xi_d|^2}}+\frac{\left(0,\xi_2,\dots,\xi_d\right)^T}{\sqrt{|\xi_1|^2+\cdots +|\xi_d|^2}}\\
    &=\alpha e^{i\theta_1}\left(1,0,\dots,0\right)^T+\sqrt{1-\alpha^2}\frac{\left(0,\xi_2,\dots,\xi_d\right)^T}{\sqrt{|\xi_2|^2+\cdots +|\xi_d|^2}}
\end{split}
\end{equation}
where $\alpha=\frac{|\xi_1|}{\sqrt{|\xi_1|^2+\cdots +|\xi_d|^2}}$ whose density is easy to calculate (see e.g. \cite[Lemma 4.2.4]{Hiai2006semicircle}). The independence between $\alpha$ and $\ket{v}=\frac{\left(0,\xi_2,\dots,\xi_d\right)^T}{\sqrt{|\xi_2|^2+\cdots +|\xi_d|^2}}$ follows from the fact that the distribution of $\ket{v}$ is independent of $\sqrt{|\xi_2|^2+\cdots +|\xi_d|^2}$. This concludes the proof.
\end{proof}

\subsection{A simple algorithm for solving DIPE}
In the following we show how to solve DIPE by performing the standard POVM. Suppose Alice has state $\ket{\phi}^{\otimes k}$ and performs the standard POVM $\left\{\binom{d+k-1}{k}\ketbra{u}^{\otimes k}\dd u\right\}$. The probability density of the measurement outcome $\ket{u}$ can be written as
\begin{equation}
    \binom{d+k-1}{k}\left|\braket{\phi}{u}\right|^{2k}\dd u
\end{equation}
where $\dd u$ denotes the density of random unit vector. This resulting distribution is certainly not uniform; it can be understood as a biased distribution towards the vector $\ket{\phi}$. Following the spirit of Lemma~\ref{lemma:decomposition}, we decompose $\ket{u}$ as
\begin{equation}
    \ket{u}=\alpha e^{i\theta}\ket{\phi}+\sqrt{1-\alpha^2}\ket{u'},
\end{equation}
where $\alpha,\theta,\ket{u'}$ are independent random variables, $\theta$ is a uniformly random phase, and $\ket{u'}$ is a random unit vector in $\mathbb{C}^{d-1}_{\phi_\perp}$, where we define $\mathbb{C}^{d-1}_{\phi_\perp}$ as the $d-1$ dimensional subspace of $\mathbb{C}^d$ that is perpendicular to $\ket{\phi}$. This looks similar to the distribution of a uniformly random unit vector; the only difference is that now $\alpha=\left|\braket{\phi}{u}\right|$ has density
\begin{equation}
    \binom{d+k-1}{k}x^{2k}\cdot 2(d-1)(1-x^2)^{d-2}x,\,\,\,\,x\in(0,1).
\end{equation}
We have already shown in Eq.~\eqref{eq:tomography} that $\E\left|\braket{\phi}{u}\right|^2=\frac{k+1}{d+k}$. This means that $\ket{u}$ and $\ket{\phi}$ have a larger overlap than the typical overlap of two random vectors. We can also calculate the variance following a similar argument as Eq.~\eqref{eq:tomography},
\begin{equation}
\begin{split}
    \Var(\left|\braket{\phi}{u}\right|^2)&=\E\left|\braket{\phi}{u}\right|^4-\left(\frac{k+1}{d+k}\right)^2\\
    &=\E_{u\sim\mathbb{C}^d}\left|\braket{\phi}{u}\right|^4\cdot\binom{d+k-1}{k}\left|\braket{\phi}{u}\right|^{2k}-\left(\frac{k+1}{d+k}\right)^2\\
        &=\binom{d+k-1}{k}\E_{u\sim\mathbb{C}^d}\Tr(\phi^{\otimes k+2}\cdot u^{\otimes k+2})-\left(\frac{k+1}{d+k}\right)^2\\
        &=\binom{d+k-1}{k}\binom{d+k+1}{k+2}^{-1}-\left(\frac{k+1}{d+k}\right)^2=\frac{(d-1) (k+1)}{(d+k)^2 (d+k+1)}.
\end{split}
\end{equation}
This means that with high probability
\begin{equation}\label{eq:standardpovmvariance}
    \left|\braket{\phi}{u}\right|^2=\frac{k}{d}\pm O\left(\frac{\sqrt{k}}{d}\right).
\end{equation}
We show that this leads to a simple algorithm for DIPE with $k=O(\sqrt{d})$ copies.

\begin{theorem}
Consider the following algorithm for DIPE which uses the template in Fig.~\ref{fig:algtemplate}a: Alice and Bob both perform the standard POVM on all copies of their states, and obtain results $\ket{u},\ket{v}$, respectively. Then they compute $\left|\braket{u}{v}\right|^2$, and decide case 2 if $\left|\braket{u}{v}\right|^2\leq\frac{10}{d}$, and case 1 otherwise. This algorithm decides DIPE with high probability when $k=O(\sqrt{d})$.
\end{theorem}
\begin{proof}
We start with case 2 where Alice has $\ket{\phi}^{\otimes k}$ and Bob has $\ket{\psi}^{\otimes k}$, for independent random vectors $\ket{\phi},\ket{\psi}\sim\mathbb{C}^d$. In this case the results $u,v$ are independent; their average overlap is easy to calculate as
\begin{equation}
\begin{split}
    \E \Tr(u v)&=\binom{d+k-1}{k}^2\E_{u,v,\phi,\psi\sim\mathbb{C}^d} \Tr(u v)\Tr(u^{\otimes k}\cdot \phi^{\otimes k})\Tr(v^{\otimes k}\cdot \psi^{\otimes k})\\
    &=\E_{u,v\sim\mathbb{C}^d} \Tr(u v)\Tr(u^{\otimes k}\cdot \psym{d}{k})\Tr(v^{\otimes k}\cdot \psym{d}{k})\\
    &=\E_{u,v\sim\mathbb{C}^d} \Tr(u v)=\frac{1}{d}.
\end{split}
\end{equation}
By Markov's inequality, we have
\begin{equation}
    \Pr[\left|\braket{u}{v}\right|^2\geq \frac{10}{d}]\leq 0.1,
\end{equation}
so the algorithm returns ``case 2" with probability at least 0.9.

In case 1 where Alice and Bob both has $\ket{\phi}^{\otimes k}$ for a random state $\ket{\phi}\sim\mathbb{C}^d$, following the above discussions we write their measurement outcomes as
$\ket{u}=\alpha e^{i\theta}\ket{\phi}+\sqrt{1-\alpha^2}\ket{u'}$ and $\ket{v}=\beta e^{i\theta'}\ket{\phi}+\sqrt{1-\beta^2}\ket{v'}$. Then we have
\begin{equation}
    \left|\braket{u}{v}\right|\geq \left|\braket{\phi}{u}\right|\cdot \left|\braket{\phi}{v}\right|-\left|\braket{u'}{v'}\right|.
\end{equation}
Note that $\ket{u'}$ and $\ket{v'}$ are independent random vectors in $\mathbb{C}^{d-1}_{\phi_\perp}$. Therefore, using Eq.~\eqref{eq:standardpovmvariance}, with high probability we have
\begin{equation}
    \begin{split}
  \left|\braket{u}{v}\right|&\geq\frac{k-O(\sqrt{k})}{d}-\frac{O(1)}{\sqrt{d}}\geq\frac{\sqrt{10}}{\sqrt{d}}
    \end{split}
\end{equation}
when $k=C\cdot\sqrt{d}$ for some large constant $C>0$. This implies that the algorithm returns ``case 1" with high probability.
\end{proof}

\subsection{Analysis of the estimation algorithm}
\begin{figure}[t]
  \begin{algorithm}[H]
    \caption{Quantum inner product estimation with multi-copy measurement}\label{alg:innerproductmulti}
    \raggedright\textbf{Input:} $k$ copies of unknown pure states $\ket{\phi},\ket{\psi}\in\mathbb{C}^{d}$\\
    \textbf{Output:} an estimate of $f=\left|\braket{\phi}{\psi}\right|^2$
    \begin{algorithmic}[1]
    \State Alice measures $\ket{\phi}^{\otimes k}$ with the standard POVM and obtains result $\ket{u}$
    \State Bob measures $\ket{\psi}^{\otimes k}$ with the standard POVM and obtains result $\ket{v}$
    \State\textbf{Return} $\frac{(d+k)^2}{k^2}\left|\braket{u}{v}\right|^2-\frac{d+2k}{k^2}$
    \end{algorithmic}
    \end{algorithm}
\end{figure}

Next we consider quantum inner product estimation ($\ipe$) with pure input states. Suppose Alice has input state $\ket{\phi}^{\otimes k}$ and Bob has input state $\ket{\psi}^{\otimes k}$, we would like to estimate $f:=\left|\braket{\phi}{\psi}\right|^2$ up to additive error $\varepsilon$, with success probability at least $2/3$. Note that unlike DIPE, here the input states are no longer random; the algorithm has to succeed with probability at least $2/3$ for all possible input states.

The estimation algorithm (presented in Algorithm~\ref{alg:innerproductmulti}) is similar to the previous section, which simply returns a function of $\left|\braket{u}{v}\right|^2$, the overlap of Alice and Bob's measurement outcomes after performing the standard POVM. However here with arbitrary input states the performance of this algorithm is much less clear, as $\left|\braket{u}{v}\right|^2$ does not necessarily have a nice correlation with $f$. As a first step, we show that the estimator given by Algorithm~\ref{alg:innerproductmulti} is unbiased.

\begin{lemma}
Let $\ket{u},\ket{v}$ be Alice and Bob's measurement outcomes in Algorithm~\ref{alg:innerproductmulti}, respectively. We have
\begin{equation}
    \E\left|\braket{u}{v}\right|^2=A+Bf
\end{equation}
where $A:=\frac{d+2k}{(d+k)^2}$, $B:=\frac{k^2}{(d+k)^2}$. Here the expectation is over the intrinsic randomness of quantum measurements.
\end{lemma}
\begin{proof}
Following the previous section, we first use the decomposition technique in Lemma~\ref{lemma:decomposition}. We can write Alice and Bob's measurement outcomes as 
\begin{equation}
    \begin{split}
        &\ket{u}=\alpha e^{i\theta}\ket{\phi}+\sqrt{1-\alpha^2}\ket{\phi'},\\
        &\ket{v}=\beta e^{i\varphi}\ket{\psi}+\sqrt{1-\beta^2}\ket{\psi'}.
    \end{split}
\end{equation}
For clarity we restate the properties of $\alpha,\theta,\ket{\phi'}$ (similar properties hold for $\beta,\varphi,\ket{\psi'}$). They are independent random variables satisfying
\begin{enumerate}
    \item $\alpha$ is a positive real variable distributed on $[0,1]$ with density
    \begin{equation}
    \binom{d+k-1}{k}x^{2k} \cdot 2(d-1)(1-x^2)^{d-2}x.
    \end{equation}
    We have already shown in the previous section that $\E\alpha^2=\frac{k+1}{d+k}$ and $\E \alpha^4=\frac{(k+2)(k+1)}{(d+k+1)(d+k)}$.
    \item $\theta$ is a uniformly random phase in $[0,2\pi]$. 
    \item $\ket{\phi'}$ is a random vector in $\mathbb{C}^{d-1}_{\phi_\perp}$, which is the $d-1$ dimensional subspace of $\mathbb{C}^d$ that is perpendicular to $\ket{\phi}$.
\end{enumerate}

We can write $\braket{u}{v}$ as
\begin{equation}
    \braket{u}{v}=\alpha\beta e^{i(\varphi-\theta)}\braket{\phi}{\psi}+\sqrt{1-\alpha^2}\beta e^{i\varphi}\braket{\phi'}{\psi}+\sqrt{1-\beta^2}\alpha e^{-i\theta}\braket{\phi}{\psi'}+\sqrt{(1-\alpha^2)(1-\beta^2)}\braket{\phi'}{\psi'}.
\end{equation}

There are 16 terms in $\left|\braket{u}{v}\right|^2$, and it's easy to see that all cross terms have a random phase, which becomes 0 after we average over the random phase. Therefore there are only 4 non-zero terms in $\E\left|\braket{u}{v}\right|^2$, given by
\begin{equation}
    \begin{split}
    \E\left|\braket{u}{v}\right|^2&=\E\left[\alpha^2\beta^2 f+(1-\alpha^2)\beta^2 \left|\braket{\phi'}{\psi}\right|^2+(1-\beta^2)\alpha^2\left|\braket{\phi}{\psi'}\right|^2+(1-\alpha^2)(1-\beta^2)\left|\braket{\phi'}{\psi'}\right|^2\right].
    \end{split}
\end{equation}
Next, note that $\E\ketbra{\phi'}$ is the maximally mixed state in $\mathbb{C}^{d-1}_{\phi_\perp}$ that can be written as $\frac{I-\ketbra{\phi}}{d-1}$, which gives
\begin{equation}
    \E\left|\braket{\phi'}{\psi}\right|^2=\frac{1-f}{d-1}.
\end{equation}
By symmetry we also have $\E\left|\braket{\phi}{\psi'}\right|^2=\frac{1-f}{d-1}$. Next, since $\phi'$ and $\psi'$ are independent, we have
\begin{equation}
\begin{split}
    \E\left|\braket{\phi'}{\psi'}\right|^2&=\E_{\phi',\psi'}\Tr(\ketbra{\phi'}\cdot \ketbra{\psi'})\\
    &=\Tr(\frac{I-\ketbra{\phi}}{d-1}\cdot \frac{I-\ketbra{\psi}}{d-1})\\
    &=\frac{d-2+f}{(d-1)^2}.
\end{split}
\end{equation}

Therefore
\begin{equation}
    \begin{split}
        \E\left|\braket{u}{v}\right|^2&=\frac{(k+1)^2}{(d+k)^2}f+\frac{2(d-1)(k+1)}{(d+k)^2}\cdot \frac{1-f}{d-1}+\frac{(d-1)^2}{(d+k)^2}\cdot\frac{d-2+f}{(d-1)^2}\\
        &=\frac{2(k+1)}{(d+k)^2}+\frac{d-2}{(d+k)^2} + f\left(\frac{(k+1)^2}{(d+k)^2}-\frac{2(k+1)}{(d+k)^2}+\frac{1}{(d+k)^2}\right)\\
        &=\frac{d+2k}{(d+k)^2}+\frac{k^2}{(d+k)^2}f.
    \end{split}
\end{equation}
\end{proof}

Let $w:=(\left|\braket{u}{v}\right|^2-A)/B$. We have shown that $\E w=f$, so $w$ is an unbiased estimator for $f$. The main question is therefore to calculate the variance of $w$.

\begin{lemma}\label{lemma:varianceofmulticopyestimation}
Let $w=(\left|\braket{u}{v}\right|^2-A)/B$ be the estimator returned by Algorithm~\ref{alg:innerproductmulti}, where $A=\frac{d+2k}{(d+k)^2}$, $B=\frac{k^2}{(d+k)^2}$. The variance of $w$ is upper bounded by
\begin{equation}
    \Var(w)\leq O\left( \frac{f}{k}+\frac{df}{k^2}+\frac{1}{k^2}+\frac{d}{k^3}+\frac{d^2}{k^4}\right).
\end{equation}
\end{lemma}

The proof involves heavy calculations and is given in Appendix~\ref{app:varianceofmulticopyestimation}. As a corollary of the above lemmas we have the following result on the performance of Algorithm~\ref{alg:innerproductmulti}.

\begin{theorem}\label{thm:multicopyalg}
For any $\varepsilon\in(0,1)$, Algorithm~\ref{alg:innerproductmulti} returns an estimate of $f$ within $\varepsilon$ additive error with probability at least $2/3$, provided that
\begin{equation}\label{eq:multicopysamplecomplexity}
    k\geq C\cdot \max\left\{\frac{f}{\varepsilon^2},\frac{\sqrt{df}}{\varepsilon},\frac{1}{\varepsilon},\frac{d^{1/3}}{\varepsilon^{2/3}},\frac{\sqrt{d}}{\sqrt{\varepsilon}}\right\}.
\end{equation}
for some constant $C>0$. 
Since $f\leq 1$, it suffices to have
\begin{equation}
    k\geq C\cdot \max\left\{\frac{1}{\varepsilon^2},\frac{\sqrt{d}}{\varepsilon}\right\},
\end{equation}
which corresponds to the sample complexity when no prior knowledge on $f$ is assumed.
\end{theorem}

\section{Lower bounds against multi-copy measurements and interactive protocols}

Next we develop lower bounds for distributed quantum inner product estimation with multi-copy measurements. We start by focusing on the decisional inner product estimation problem (DIPE, Definition~\ref{def:DIPE}). In Section~\ref{sec:interactivelowerbound} we prove a $\Omega(\sqrt{d})$ lower bound for DIPE against arbitrary interactive protocols with multi-copy measurements. The use of optimal cloning in the lower bound (as shown below) eludes its intuitive understanding, so we attempt to find a more elementary argument in Section~\ref{sec:onewaylowerbound}. The argument currently works in the special case where Alice performs the standard POVM and sends the result to Bob, who then performs a two-outcome POVM to solve DIPE. We show that this protocol has a nice structure which allows us to exactly identify Bob's optimal measurement strategy after receiving Alice's message. Next, in Section~\ref{sec:estimationlowerbound} we construct a lower bound instance for estimating the inner product within $\varepsilon$ accuracy, and prove a $\Omega(\sqrt{d}/\varepsilon)$ lower bound based on a reduction to DIPE. Finally, in Section~\ref{sec:swaptest} we discuss the SWAP test and prove its optimality via a $\Omega(1/\varepsilon^2)$ lower bound for inner product estimation in the non-distributed setting. Combining the two lower bounds implies the lower bound in our main result.

Recall that in the DIPE problem, Alice and Bob are required to distinguish between (1) they both have $k$ copies of the same random state $\ket{\phi}$, and (2) Alice has $\ket{\phi}^{\otimes k}$ and Bob has $\ket{\psi}^{\otimes k}$, where $\ket{\phi}$ and $\ket{\psi}$ are independent random states. Alice and Bob's protocol can be viewed as implementing a two-outcome POVM $\{M,I-M\}$ that acts on $(\mathbb{C}^d)^{\otimes 2k}$, where outcome $M$ corresponds to ``case 2". In order to solve the DIPE problem, $M$ needs to satisfy
\begin{enumerate}
\item $\Pr_{\phi\sim\mathbb{C}^d}\left[M \text{ accepts }\phi^{\otimes k}\otimes \phi^{\otimes k}\right]=\E_{\phi\sim\mathbb{C}^d}\Tr\left(M \phi^{\otimes k}\otimes \phi^{\otimes k}\right)\leq\frac{1}{3}$,
     \item $\Pr_{\phi,\psi\sim\mathbb{C}^d}\left[M \text{ accepts }\phi^{\otimes k}\otimes \psi^{\otimes k}\right]=\E_{\phi,\psi\sim\mathbb{C}^d}\Tr\left(M \phi^{\otimes k}\otimes \psi^{\otimes k}\right)\geq\frac{2}{3}$.
\end{enumerate}
It's easy to see that the above requirements can be satisfied if and only if $M$ satisfies 
\begin{equation}\label{eq:dipedistinguish}
    \E_{\phi,\psi\sim\mathbb{C}^d}\Tr\left(M \left(\phi^{\otimes k}\otimes \psi^{\otimes k}-\phi^{\otimes k}\otimes \phi^{\otimes k}\right)\right)\geq\frac{1}{3}.
\end{equation}

When Alice and Bob are allowed general quantum measurements on $(\mathbb{C}^d)^{\otimes 2k}$, the above task can be solved with $k=O(1)$ using SWAP test (see Section~\ref{sec:swaptest} for more details). However, here they are only allowed local quantum operations and interactive classical communication. In the following we establish lower bounds for $k$ using these measurements.

\subsection{Lower bound against arbitrary interactive protocols}
\label{sec:interactivelowerbound}

We start by considering the general one-way protocol: suppose Alice performs an arbitrary POVM $\{M_i\}$ and sends the result $i$ to Bob, and after seeing $i$ Bob performs a two-outcome POVM $\{N_i,I-N_i\}$ where $N_i$ corresponds to ``case 2". Here for simplicity we consider discrete POVMs for Alice; the same argument also holds for continuous POVMs. The measurement operator implemented by Alice and Bob that corresponds to ``case 2" is 
\begin{equation}
    M_{AB}=\sum_i M_i\otimes N_i.
\end{equation}
From Eq.~\eqref{eq:dipedistinguish} we have
\begin{equation}
    \begin{split}
        \frac{1}{3}&\leq \E_{\phi,\psi\sim\mathbb{C}^d}\Tr\left(M_{AB} \left(\phi^{\otimes k}\otimes \psi^{\otimes k}-\phi^{\otimes k}\otimes \phi^{\otimes k}\right)\right)\\
        &=\Tr\left(M_{AB} \left(\frac{\psym{d}{k}}{\binom{d+k-1}{k}}\otimes \mmix-\E_{\phi\sim\mathbb{C}^d}\phi^{\otimes k}\otimes \phi^{\otimes k}\right)\right)\\
        &=\sum_i \frac{\Tr(M_i\psym{d}{k})}{\binom{d+k-1}{k}}\Tr(N_i(\mmix-\rho_i)),
    \end{split}
\end{equation}
where $\rho_i$ is Bob's effective post-measurement state in ``case 1" after seeing $i$,
\begin{equation}
    \rho_i:=\frac{\binom{d+k-1}{k}}{\Tr(M_i\psym{d}{k})}\E_{\phi\sim\mathbb{C}^d}\Tr(M_i \ketbra{\phi}^{\otimes k})\ketbra{\phi}^{\otimes k}.
\end{equation}
In the following we show that $\rho_i$ is always close to $\mmix$ in max-relative entropy for an arbitrary $M_i$ when $k=o(\sqrt{d})$, which implies that Bob cannot distinguish which case they are in.

To achieve this we build a connection between DIPE and optimal cloning, a well studied task in quantum information. Note that Bob's post-measurement state in ``case 1" can be written as
\begin{equation}
    \rho_i=\frac{\binom{d+k-1}{k}}{\Tr(\psym{d}{k}M_i\psym{d}{k})}\E_{\phi\sim\mathbb{C}^d}\Tr(\psym{d}{k}M_i\psym{d}{k} \ketbra{\phi}^{\otimes k})\ketbra{\phi}^{\otimes k}.
\end{equation}
Define the following operator which can be viewed as a density matrix on $\ssym{d}{k}$,
\begin{equation}
    \tau:=\frac{\psym{d}{k}M_i\psym{d}{k}}{\Tr(\psym{d}{k}M_i\psym{d}{k})}
\end{equation}
and the ``measure-and-prepare" channel
\begin{equation}
    \mathrm{MP}(\rho):=\binom{d+k-1}{k}\E_{\phi\sim\mathbb{C}^d}\Tr(\rho\cdot \ketbra{\phi}^{\otimes k})\ketbra{\phi}^{\otimes k},
\end{equation}
then Bob's state can be viewed as the output of $\mathrm{MP}$ acting on the measurement operator, $\rho_{\mathrm{Bob}}=\mathrm{MP}(\tau)$. Analogous to the fact that the standard POVM is the optimal POVM for tomography, it is shown that $\mathrm{MP}$ is the optimal measure-and-prepare channel that maps states of the form $\ket{\psi}^{\otimes k}$ to an approximation of $\ket{\psi}^{\otimes k}$~\cite{harrow2013church}. Our goal here is to show that when $k$ is small, the output state of $\mathrm{MP}$ is always close to the maximally mixed state $\mmix$, regardless of the input.

The cloning channel is defined in a somewhat similar spirit:
\begin{equation}
    \mathrm{Clone}_{s\to k}(\rho):=\frac{\binom{d+s-1}{s}}{\binom{d+k-1}{k}}\psym{d}{k}\left(\rho\otimes I^{\otimes k-s}\right)\psym{d}{k}.
\end{equation}
Here the input state $\rho$ is supported on $(\mathbb{C}^d)^{\otimes s}$ and output is supported on $(\mathbb{C}^d)^{\otimes k}$. This channel is the optimal channel that maps states of the form $\ket{\psi}^{\otimes s}$ to an approximation of $\ket{\psi}^{\otimes k}$. When $s=0$ this channel prepares the maximally mixed state $\mmix$. Chiribella showed a remarkable connection between the above two channels.

\begin{lemma}[Chiribella~\cite{Chiribella2011on}, also see~\cite{harrow2013church}]
The optimal measure-and-prepare channel is a convex combination of optimal cloning channels: for any density matrix $\rho$ on $\ssym{d}{k}$,
\begin{equation}
    \mathrm{MP}(\rho)=\sum_{s=0}^k \frac{\binom{k}{s}\binom{d+k-1}{k-s}}{\binom{d+2k-1}{k}}\mathrm{Clone}_{s\to k}(\Tr_{k-s}\rho).
\end{equation}
\end{lemma}

Intuitively, when $k$ is small the performance of $\mathrm{MP}$ should be bad, and therefore the output should contain little information about the input and be close to $\mmix$. Chiribella's theorem can be used to formally establish this closeness in max-relative entropy.

\begin{lemma}\label{lemma:dmaxcloseness}
Let $M$ be an arbitrary positive semi-definite Hermitian operator on $(\mathbb{C}^d)^{\otimes k}$. Then Bob's post-measurement state corresponding to $M$ satisfies
\begin{equation}
    \rho_{\mathrm{Bob}}=\frac{\binom{d+k-1}{k}}{\Tr(M\psym{d}{k})}\E_{\phi\sim\mathbb{C}^d}\Tr(M \ketbra{\phi}^{\otimes k})\ketbra{\phi}^{\otimes k}\geq e^{-k^2/d}\mmix.
\end{equation}
\end{lemma}
\begin{proof}
Let $\tau=\frac{\psym{d}{k}M\psym{d}{k}}{\Tr(\psym{d}{k}M\psym{d}{k})}$ be a density matrix on $\ssym{d}{k}$. Then
\begin{equation}
    \begin{split}
        \rho_{\mathrm{Bob}}&=\mathrm{MP}(\tau)\\
        &=\sum_{s=0}^k \frac{\binom{k}{s}\binom{d+k-1}{k-s}}{\binom{d+2k-1}{k}}\mathrm{Clone}_{s\to k}(\Tr_{k-s}\tau)\\
        &\geq \frac{\binom{d+k-1}{k}}{\binom{d+2k-1}{k}}\mmix=\frac{(d+k-1)\cdots d}{(d+2k-1)\cdots (d+k)}\mmix\geq e^{-k^2/d}\mmix.
    \end{split}
\end{equation}
Here in the first inequality we simply discard all the (positive) terms that correspond to $s\neq 0$.
\end{proof}

It's easy to see that Lemma~\ref{lemma:dmaxcloseness} implies $\Omega(\sqrt{d})$ lower bound against arbitrary one-way protocols. Below we show that it also provides the key component for the lower bound against arbitrary interactive protocols. In fact we prove a stronger result, which works against arbitrary separable measurements.

\begin{definition}[Separable measurements]
A two-outcome POVM $\{M,I-M\}$ is separable, denoted as $\{M,I-M\}\in\mathrm{SEP}$, if and only if it satisfies
\begin{equation}
    M=\sum_t A_t\otimes B_t,\,\,\,\,I-M=\sum_t A_t'\otimes B_t',\,\,\,\,A_t,B_t,A_t',B_t'\geq 0.
\end{equation}
Here $A_t,A_t'$ are Hermitian operators supported on Alice's system and $B_t,B_t'$ are Hermitian operators supported on Bob's system.
\end{definition}

It is well-known that separable operations strictly contains LOCC operations -- that can be implemented by an interactive protocol~\cite{Bennett1999quantum,Chitambar2009Nonlocal}, and it is a common practice to use SEP as a larger approximation of LOCC due to its simple form. In the following we prove a $\Omega(\sqrt{d})$ lower bound against SEP.

\begin{theorem}
$k=\Omega(\sqrt{d})$ copies is necessary for Alice and Bob to solve DIPE, even when they are allowed arbitrary interactive protocols (or arbitrary LOCC operations), or more generally, arbitrary separable operations.
\end{theorem}
\begin{proof}
Let $\{M,I-M\}\in\mathrm{SEP}$, where $M=\sum_t A_t\otimes B_t$ corresponds to ``case 2". For Alice and Bob to successfully solve DIPE, $M$ satisfies
\begin{equation}
    \begin{split}
        \frac{1}{3}&\leq \E_{\phi,\psi\sim\mathbb{C}^d}\Tr\left(M \left(\phi^{\otimes k}\otimes \psi^{\otimes k}-\phi^{\otimes k}\otimes \phi^{\otimes k}\right)\right)\\
        &=\sum_t\E_{\phi,\psi\sim\mathbb{C}^d}\Tr(A_t \phi^{\otimes k})\Tr(B_t \psi^{\otimes k})-\E_{\phi\sim\mathbb{C}^d}\Tr(M\left(\phi^{\otimes k}\otimes \phi^{\otimes k}\right))\\
        &=\sum_t\frac{\Tr(A_t \psym{d}{k})}{\binom{d+k-1}{k}}\Tr(B_t \mmix)-\E_{\phi\sim\mathbb{C}^d}\Tr(M\left(\phi^{\otimes k}\otimes \phi^{\otimes k}\right)).
    \end{split}
\end{equation}
Using Lemma~\ref{lemma:dmaxcloseness} we have
\begin{equation}
    \mmix\leq e^{k^2/d} \frac{\binom{d+k-1}{k}}{\Tr(A_t\psym{d}{k})}\E_{\phi\sim\mathbb{C}^d}\Tr(A_t \ketbra{\phi}^{\otimes k})\ketbra{\phi}^{\otimes k},
\end{equation}
which implies that
\begin{equation}
\begin{split}
    \Tr(B_t\mmix)&\leq e^{k^2/d} \frac{\binom{d+k-1}{k}}{\Tr(A_t\psym{d}{k})}\E_{\phi\sim\mathbb{C}^d}\Tr(A_t \ketbra{\phi}^{\otimes k})\Tr(B_t\ketbra{\phi}^{\otimes k})\\
    &=e^{k^2/d} \frac{\binom{d+k-1}{k}}{\Tr(A_t\psym{d}{k})}\E_{\phi\sim\mathbb{C}^d}\Tr((A_t\otimes B_t)\phi^{\otimes k}\otimes \phi^{\otimes k}).
\end{split}
\end{equation}
Therefore,
\begin{equation}
    \begin{split}
        \frac{1}{3}&\leq \sum_t\frac{\Tr(A_t \psym{d}{k})}{\binom{d+k-1}{k}}\Tr(B_t \mmix)-\E_{\phi\sim\mathbb{C}^d}\Tr(M\left(\phi^{\otimes k}\otimes \phi^{\otimes k}\right))\\
        &\leq \sum_t e^{k^2/d} \E_{\phi\sim\mathbb{C}^d}\Tr((A_t\otimes B_t)\phi^{\otimes k}\otimes \phi^{\otimes k})-\E_{\phi\sim\mathbb{C}^d}\Tr(M\left(\phi^{\otimes k}\otimes \phi^{\otimes k}\right))\\
        &=\left(e^{k^2/d}-1\right)\E_{\phi\sim\mathbb{C}^d}\Tr(M\left(\phi^{\otimes k}\otimes \phi^{\otimes k}\right))\\
        &\leq e^{k^2/d}-1,
    \end{split}
\end{equation}
which implies that $k=\Omega(\sqrt{d})$.
\end{proof}

\subsection{Understanding a one-way protocol for DIPE}
\label{sec:onewaylowerbound}

Next we analyze the following one-way protocol which is a natural generalization of the algorithms presented in Section~\ref{sec:multicopy} (restated for convenience):
\begin{enumerate}
    \item Alice measures all copies of her state with the standard POVM, gets result $u$, and sends to Bob.
    \item Conditioned on $u$, Bob performs a global two-outcome measurement $\{M_u,\psym{d}{k}-M_u\}$ on all copies of his state, and decides ``case 2" if he sees $M_u$, and ``case 1" otherwise.
\end{enumerate}
The reason for presenting this analysis is the following utility compared with the previous section. First, here we are able to exactly identify Bob's optimal measurement strategy for the above protocol, which provides intuition on how to think about the problem in an interactive setting. Second, the analysis implies a $\Omega(\sqrt{d})$ lower bound for this protocol without using Chiribella's theorem.

Let $M_{AB}$ be the measurement operator implemented by Alice and Bob that corresponds to ``case 2". Then
\begin{equation}
    M_{AB}=\E_{u\sim\mathbb{C}^d}\binom{d+k-1}{k}\ketbra{u}^{\otimes k}\otimes M_u.
\end{equation}
Following the discussion above, $M_{AB}$ satisfies
\begin{equation}\label{eq:onewaydistinguish}
    \begin{split}
    \frac{1}{3}&\leq \E_{\phi,\psi\sim\mathbb{C}^d}\Tr\left(M_{AB} \left(\phi^{\otimes k}\otimes \psi^{\otimes k}-\phi^{\otimes k}\otimes \phi^{\otimes k}\right)\right)\\
    &=\E_{\phi,u\sim\mathbb{C}^d}\Tr\left(M_u \mmix\right)-\binom{d+k-1}{k}\Tr(u^{\otimes k}\phi^{\otimes k})\cdot \Tr(M_u \phi^{\otimes k})\\
    &=\E_{u\sim\mathbb{C}^d}\Tr(M_u(\mmix-\rho_u)).
    \end{split}
\end{equation}
Here, $\mmix=\frac{1}{\binom{d+k-1}{k}}\psym{d}{k}$ is the maximally mixed state in $\ssym{d}{k}$, and 
\begin{equation}
\label{eq:Bobstate}
    \rho_u:=\binom{d+k-1}{k}\E_{\phi\sim\mathbb{C}^d}\ketbra{\phi}^{\otimes k}\left|\braket{u}{\phi}\right|^{2k}
\end{equation}
is Bob's post-measurement state in ``case 1" after seeing $u$. Identifying Bob's optimal measurement is equivalent to computing the trace distance between $\rho_u$ and $\mmix$, as
\begin{equation}\label{eq:tracedistance2}
    \max_{0\leq M_u\leq \psym{d}{k}}\Tr(M_u(\mmix-\rho_u))=\frac{1}{2}\|\rho_u-\mmix\|_1.
\end{equation}
Therefore we need to analyze the spectrum of $\rho_u$ and compare with the uniform distribution $\mmix$.

We show that the state $\rho_u$ has a nice block-diagonal structure. For any fixed state $\ket{u}\in\mathbb{C}^d$, define the following subspaces
\begin{equation}
    W_u^t=\mathrm{span}\left\{\psym{d}{k}\ket{u}^{\otimes t}\otimes \ket{v}^{\otimes k-t}:\ket{v}\in\mathbb{C}^d, \braket{v}{u}=0\right\},
\end{equation}
and let $\Pi_u^t$ be the orthogonal projector onto $W_u^t$. Clearly, each $W_u^t$ ($t=0,\dots,k$) is a subspace of $\ssym{d}{k}$, as we are applying the projector $\psym{d}{k}$ on all states in the definition, and the subspaces $W_u^t$ are mutually orthogonal. In fact, we show that the symmetric subspace $\ssym{d}{k}$ can be decomposed as a direct sum of $W_u^t$ as
\begin{equation}
    \ssym{d}{k}=W_u^0\oplus \cdots\oplus W_u^k.
\end{equation}
Let's denote $\ket{u}$ as $\ket{0}$. To prove this, we use a characterization of $\ssym{d}{k}$ (Definition~\ref{def:symmetricsubspace}) as
\begin{equation}
    \begin{split}
    \ssym{d}{k}&=\vspan\left\{\sum_{i:T(i)=\ell}\ket{i_1,\dots,i_k}:\ell=(\ell_0,\dots,\ell_{d-1}),\ell_j\geq 0,\sum_{j=0}^{d-1}\ell_j=k\right\}\\
    &=\bigoplus_{t=0}^k\vspan\left\{\sum_{i:T(i)=\ell}\ket{i_1,\dots,i_k}:\ell=(\ell_0,\dots,\ell_{d-1}),\ell_j\geq 0,\ell_0=t,\sum_{j=0}^{d-1}\ell_j=k\right\}\\
    &=\bigoplus_{t=0}^k W_u^t.
    \end{split}
\end{equation}
Using the above characterization it is easy to show that $\mathrm{dim}W_u^t=\binom{d+k-t-2}{k-t}$.
We show that Bob's state $\rho_u$ is block-diagonal in the subspaces $W_u^t$, and furthermore is uniform in each of the subspaces. Therefore $\rho_u=\sum_t \beta_t \Pi_u^t$ with coefficients $\beta_t>0$. 

\begin{lemma}\label{lemma:blockdiagonal}
For any $\ket{u}\in\mathbb{C}^d$, Bob's post-measurement state $\rho_u$ (Eq.~\eqref{eq:Bobstate}) has the following block-diagonal form
\begin{equation}
    \rho_u=\sum_{t=0}^k \beta_t \Pi_u^t,\,\,\,\,\beta_t=\frac{(k+t)(k+t-1)\cdots (t+1)}{(d+2k-1)(d+2k-2)\cdots (d+k)}.
\end{equation}
\end{lemma}
The proof is presented in Appendix~\ref{app:blockdiagonal}. Meanwhile, the maximally mixed state can be written as
\begin{equation}
    \mmix=\frac{1}{\binom{d+k-1}{k}}\psym{d}{k}=\frac{1}{\binom{d+k-1}{k}}\sum_{t=0}^k \Pi_u^t.
\end{equation}
As the above decompositions for $\rho_u$ and $\mmix$ have the same block-diagonal structure, it's easy to calculate the trace distance between $\rho_u$ and $\mmix$ as well as the optimal measurement to distinguish them. 

\begin{lemma}
\label{lemma:tracedistance}
Suppose $k=O(\sqrt{d})$ and $k=\omega(1)$, then the optimal POVM to distinguish $\rho_u$ and $\mmix$ is given by $\{\Pi_u^0, \psym{d}{k}-\Pi_u^0\}$, and we have
\begin{equation}\label{eq:optimalmeasurement}
    \frac{1}{2}\left\|\rho_u-\mmix\right\|_1=\Tr(\Pi_u^0 (\mmix-\rho_u))=\frac{d-1}{d+k-1}-\frac{(d+k-2)\cdots(d-1)}{(d+2k-1)\cdots(d+k)}.
\end{equation}
\end{lemma}
\begin{proof}
First note that for any $m\geq 0$, $\frac{d+m}{d+k+m}\geq\frac{d}{d+k}$. This gives
\begin{equation}
    \frac{(d+k-1)\cdots d}{(d+2k-1)\cdots (d+k)}\geq \left(\frac{d}{d+k}\right)^k\geq e^{-k^2/d}.
\end{equation}

To construct the optimal distinguishing measurement, we need to compare $\beta_t$ with $\binom{d+k-1}{k}^{-1}$. For each $t=0,\dots,k$, we have
\begin{equation}\label{eq:comparespectrum}
    \beta_t\binom{d+k-1}{k}=\binom{k+t}{k}\frac{(d+k-1)\cdots d}{(d+2k-1)\cdots (d+k)}.
\end{equation}
When $t=0$, this shows that $\beta_0<\binom{d+k-1}{k}^{-1}$. When $t\geq 1$, we have
\begin{equation}
    \beta_t\binom{d+k-1}{k}\geq k\cdot \frac{(d+k-1)\cdots d}{(d+2k-1)\cdots (d+k)}\geq k\cdot e^{-k^2/d}=\omega(1).
\end{equation}
Therefore, when $d$ is large enough, $W_u^0$ is the only subspace on which $\mmix$ has larger support than $\rho_u$. This proves the first equality in Eq.~\eqref{eq:optimalmeasurement}. For the second equality, we have
\begin{equation}
    \Tr(\Pi_u^0 \mmix)=\frac{\mathrm{dim}W_u^0}{\binom{d+k-1}{k}}=\frac{\binom{d+k-2}{k}}{\binom{d+k-1}{k}}=\frac{d-1}{d+k-1},
\end{equation}
and
\begin{equation}
    \Tr(\Pi_u^0 \rho_u)=\beta_0\cdot \mathrm{dim}W_u^0=\frac{(d+k-2)\cdots(d-1)}{(d+2k-1)\cdots(d+k)}.
\end{equation}
\end{proof}

Lemma~\ref{lemma:tracedistance} implies the optimal strategy for Bob's measurement as well as a $\Omega(\sqrt{d})$ lower bound for this protocol.

\begin{theorem}
Consider a one-way protocol where Alice performs the standard measurement and sends the result to Bob. Then Bob's optimal measurement strategy conditioned on receiving $u$ is the two-outcome POVM $\{\Pi_u^0, \psym{d}{k}-\Pi_u^0\}$. Furthermore, this protocol requires $k=\Omega(\sqrt{d})$ copies for Alice and Bob to solve DIPE.
\end{theorem}
\begin{proof}
Without loss of generality we assume that $k=O(\sqrt{d})$ and $k=\omega(1)$, as a lower bound that works against $k=\omega(1)$ will also work against $k=O(1)$.

Following the discussion around Eq.~\eqref{eq:onewaydistinguish}, Alice and Bob aims to maximize $\E_{u\sim\mathbb{C}^d}\Tr(M_u(\mmix-\rho_u))$. For any $\ket{u}\in\mathbb{C}^d$ we have identified the measurement that maximizes $\Tr(M_u(\mmix-\rho_u))$. Therefore the overall optimal strategy for Bob is to measure $\{\Pi_u^0, \psym{d}{k}-\Pi_u^0\}$ after seeing Alice's outcome $u$. In order for Alice and Bob to sucessfully distinguish with probability $2/3$, we have
\begin{equation}
    \begin{split}
        \frac{1}{3}&\leq \E_{u\sim\mathbb{C}^d}\Tr(M_u(\mmix-\rho_u))\\
        &\leq \E_{u\sim\mathbb{C}^d}\Tr(\Pi_u^0(\mmix-\rho_u))\\
        &=\frac{d-1}{d+k-1}-\frac{(d+k-2)\cdots(d-1)}{(d+2k-1)\cdots(d+k)}\\
        &=O(k^2/d).
    \end{split}
\end{equation}
Therefore $k=\Omega(\sqrt{d})$ is necessary for Alice and Bob to solve DIPE.
\end{proof}

\subsection{Lower bound for estimation}
\label{sec:estimationlowerbound}
Next we extend our lower bound for DIPE to a lower bound for $\ipe$. The key step is to consider the following decision problem as a lower bound instance.

\begin{problem}
\label{prob:lowerbounddistinguish}
Suppose Alice and Bob are each given $k$ copies of a pure state in $\mathbb{C}^{d+1}$. Let $\varepsilon\in(0,1)$. They are promised that one of the following two cases hold:
\begin{enumerate}
    \item Alice has $\left(\sqrt{1-\varepsilon}e^{i\theta}\ket{0}+\sqrt{\varepsilon}\ket{\phi}\right)^{\otimes k}$ and Bob has $\left(\sqrt{1-\varepsilon}e^{i\theta'}\ket{0}+\sqrt{\varepsilon}\ket{\phi}\right)^{\otimes k}$, where $\ket{\phi}$ is a uniformly random state in  $\vspan\{\ket{1},\ket{2},\dots,\ket{d}\}$, and $\theta,\theta'$ are independent random phases in $[0,2\pi]$.
    \item Alice has $\left(\sqrt{1-\varepsilon}e^{i\theta}\ket{0}+\sqrt{\varepsilon}\ket{\phi}\right)^{\otimes k}$ and Bob has $\left(\sqrt{1-\varepsilon}e^{i\theta'}\ket{0}+\sqrt{\varepsilon}\ket{\psi}\right)^{\otimes k}$, where $\ket{\phi},\ket{\psi}$ are independent random states in  $\vspan\{\ket{1},\ket{2},\dots,\ket{d}\}$, and $\theta,\theta'$ are independent random phases in $[0,2\pi]$.
\end{enumerate}
Their goal is to decide which case they are in with success probability at least $2/3$, using an interactive protocol that involves local quantum operations and classical communication.
\end{problem}

Note that in the above definition, the random states $\ket{\phi},\ket{\psi}$ are supported on $\vspan\{\ket{1},\ket{2},\dots,\ket{d}\}$ and are always orthogonal to $\ket{0}$. We have increased the dimension from $d$ to $d+1$ for convenience, which does not change the form of our lower bound. Our main result in this section is the following lower bound for this problem.

\begin{theorem}
\label{thm:epsdecisionlowerbound}
$k=\Omega(\sqrt{d}/\varepsilon)$ copies are necessary for Alice and Bob to solve Problem~\ref{prob:lowerbounddistinguish}, even when they are allowed arbitrary interactive protocols (or arbitrary LOCC operations), or more generally, arbitrary separable operations.
\end{theorem}

Note that here (as well as in Theorem~\ref{thm:epsestimationlowerbound}) the lower bound holds for arbitrary $\varepsilon\in(0,1)$. The proof follows from a reduction to DIPE and is presented later. We show that this implies a tight lower bound for $\ipe$.

\begin{theorem}
\label{thm:epsestimationlowerbound}
Suppose Alice has input $\ket{\phi}^{\otimes k}$ and Bob has input $\ket{\psi}^{\otimes k}$, for arbitrary unknown pure states $\ket{\phi},\ket{\psi}\in\mathbb{C}^{d}$. Then $k=\Omega(\sqrt{d}/\varepsilon)$ copies is necessary for them to estimate $\Tr(\phi\psi)=\left|\braket{\phi}{\psi}\right|^2$ up to additive error $\varepsilon$ with success probability $2/3$, even when they are allowed arbitrary interactive protocols (or arbitrary LOCC operations), or more generally, arbitrary separable operations.
\end{theorem}

\begin{remark}
In Section~\ref{sec:swaptest} we present a $\Omega(1/\varepsilon^2)$ lower bound for $\ipe$, which is proven in the non-distributed setting and also applies to our distributed setting. Combining the two lower bounds gives the lower bound in our main result. See Section~\ref{sec:swaptest} for more details.
\end{remark}

\begin{proof}
Note that for any estimation algorithm, its success probability can be amplified to $1-\delta$ by taking the median of $O(\log\frac{1}{\delta})$ independent executions. Therefore choosing the success probability to be an arbitrary constant will only change the sample complexity by a constant factor. Suppose there is an algorithm $\mc A$ with interactive communication between Alice and Bob that, on input $k$ copies of arbitrary unknown pure states, estimates the inner product up to additive error $\varepsilon/100$, with success probability at least 0.99. Here $\varepsilon$ is the parameter in the definition of Problem~\ref{prob:lowerbounddistinguish}, and we assume that $\varepsilon\leq 0.01$. To prove the theorem statement it suffices to prove that using this protocol they can solve Problem~\ref{prob:lowerbounddistinguish} with success probability at least $2/3$.

Consider the inner product between Alice and Bob's states of the two cases in Problem~\ref{prob:lowerbounddistinguish}. In case 1 the inner product of Alice and Bob's input states is given by
\begin{equation}
    f_1:=\left|(1-\varepsilon)e^{i(\theta'-\theta)}+\varepsilon\right|^2=1-2\varepsilon +2\varepsilon^2+2\varepsilon(1-\varepsilon)\cos(\theta'-\theta).
\end{equation}
In case 2 the inner product is given by
\begin{equation}
    f_2:=\left|(1-\varepsilon)e^{i(\theta'-\theta)}+\varepsilon\braket{\phi}{\psi}\right|^2=(1-\varepsilon)^2+\varepsilon^2\left|\braket{\phi}{\psi}\right|^2+2\varepsilon(1-\varepsilon)\Re(e^{i(\theta'-\theta)}\braket{\psi}{\phi}).
\end{equation}
A sketch of the remaining argument is as follows: in case 1 the inner product $f_1$ is distributed roughly as $f_1=1-2\varepsilon+2\varepsilon\cos(\gamma)$ for a uniformly random phase $\gamma$. In case 2 the inner product $f_2$ is concentrated around $(1-\varepsilon)^2$, as $\left|\braket{\psi}{\phi}\right|= O(1/\sqrt{d})$ with high probability. Therefore, after Alice and Bob estimate the inner product between their states using algorithm $\mc A$, they can decide case 2 if the result is close to $(1-\varepsilon)^2$, and decide case 1 otherwise.\\

\noindent\textbf{Algorithm $\mc B$ for Problem~\ref{prob:lowerbounddistinguish}:}
\begin{enumerate}
    \item Alice and Bob estimate the inner product between their states using algorithm $\mc A$, obtain result $f$.
    \item They decide case 2 if $f\in[(1-\varepsilon)^2-\frac{\varepsilon}{50},(1-\varepsilon)^2+\frac{\varepsilon}{50}]$, and decide case 1 otherwise.
\end{enumerate}
The performance of algorithm $\mc A$ guarantees that if Alice and Bob are in case $i$ ($i=1,2$), then $|f-f_i|\leq\frac{\varepsilon}{100}$ holds with probability at least 0.99. In the following we show that algorithm $\mc B$ can decide Problem~\ref{prob:lowerbounddistinguish} with success probability at least $2/3$.

If Alice and Bob are in case 1, then
\begin{equation}
\begin{split}
    &\Pr[f\in[(1-\varepsilon)^2-\frac{\varepsilon}{50},(1-\varepsilon)^2+\frac{\varepsilon}{50}]]\\
    &\leq \Pr[f\in[(1-\varepsilon)^2-\frac{\varepsilon}{50},(1-\varepsilon)^2+\frac{\varepsilon}{50}]\wedge|f-f_1|\leq\frac{\varepsilon}{100}]+0.01\\
    &\leq \Pr[f_1\in[(1-\varepsilon)^2-\frac{3\varepsilon}{100},(1-\varepsilon)^2+\frac{3\varepsilon}{100}]]+0.01\\
    &\leq \Pr_{\gamma\sim[0,2\pi]}[\cos(\gamma)\in[-0.0203,0.0152]]+0.01\\
    &\leq 0.03.
\end{split}
\end{equation}
Therefore the algorithm returns case 1 with probability at least 0.97.

If Alice and Bob are in case 2, first note that 
\begin{equation}
    \begin{split}
        \left|\varepsilon^2\left|\braket{\phi}{\psi}\right|^2+2\varepsilon(1-\varepsilon)\Re(e^{i(\theta'-\theta)}\braket{\psi}{\phi})\right|\leq \varepsilon^2\left|\braket{\phi}{\psi}\right|+2\varepsilon(1-\varepsilon)\left|\braket{\psi}{\phi}\right|&\leq 2\varepsilon \left|\braket{\psi}{\phi}\right|.
    \end{split}
\end{equation}
Therefore,
\begin{equation}
\begin{split}
    &\Pr[f\in[(1-\varepsilon)^2-\frac{\varepsilon}{50},(1-\varepsilon)^2+\frac{\varepsilon}{50}]]\\
    &\geq \Pr[f_2\in[(1-\varepsilon)^2-\frac{\varepsilon}{100},(1-\varepsilon)^2+\frac{\varepsilon}{100}]\wedge |f-f_2|\leq\frac{\varepsilon}{100}]\\
    &\geq \Pr[f_2\in[(1-\varepsilon)^2-\frac{\varepsilon}{100},(1-\varepsilon)^2+\frac{\varepsilon}{100}]]+\Pr[|f-f_2|\leq\frac{\varepsilon}{100}]-1\\
    &\geq \Pr[f_2\in[(1-\varepsilon)^2-\frac{\varepsilon}{100},(1-\varepsilon)^2+\frac{\varepsilon}{100}]]-0.01\\
    &\geq \Pr_{\phi,\psi\sim\mathbb{C}^d}[\left|\braket{\psi}{\phi}\right|\leq\frac{1}{200}]-0.01.
\end{split}
\end{equation}
When $d$ is larger than some fixed constant, $\left|\braket{\psi}{\phi}\right|\leq\frac{1}{200}$ holds with high probability. Therefore the algorithm returns case 2 with high probability.
\end{proof}

\noindent\textbf{Proof of Theorem~\ref{thm:epsdecisionlowerbound}.} The main idea of the proof is to leverage the independent random phases $\theta,\theta'$ to build a reduction to DIPE.

Suppose Alice and Bob implement a two-outcome POVM $\{M,I-M\}$ via an interactive protocol where $M$ corresponds to ``case 2". In order to solve Problem~\ref{prob:lowerbounddistinguish}, $M$ satisfies
\begin{equation}
    \frac{1}{3}\leq \E_{\substack{\phi,\psi\sim\mathbb{C}^d\\\theta,\theta'\sim[0,2\pi]}}\Tr(M(\Theta_{AB}-\Psi_{AB})),
\end{equation}
where
\begin{equation}
\begin{split}
    \Theta_{AB}=&\left((1-\varepsilon)\ketbra{0}+\sqrt{\varepsilon(1-\varepsilon)}e^{i\theta}\ketbra{0}{\phi}+\sqrt{\varepsilon(1-\varepsilon)}e^{-i\theta}\ketbra{\phi}{0}+\varepsilon\ketbra{\phi}\right)^{\otimes k}\\
    &\otimes \left((1-\varepsilon)\ketbra{0}+\sqrt{\varepsilon(1-\varepsilon)}e^{i\theta'}\ketbra{0}{\psi}+\sqrt{\varepsilon(1-\varepsilon)}e^{-i\theta'}\ketbra{\psi}{0}+\varepsilon\ketbra{\psi}\right)^{\otimes k}
\end{split}
\end{equation}
and
\begin{equation}
\begin{split}
    \Psi_{AB}=&\left((1-\varepsilon)\ketbra{0}+\sqrt{\varepsilon(1-\varepsilon)}e^{i\theta}\ketbra{0}{\phi}+\sqrt{\varepsilon(1-\varepsilon)}e^{-i\theta}\ketbra{\phi}{0}+\varepsilon\ketbra{\phi}\right)^{\otimes k}\\
    &\otimes \left((1-\varepsilon)\ketbra{0}+\sqrt{\varepsilon(1-\varepsilon)}e^{i\theta'}\ketbra{0}{\phi}+\sqrt{\varepsilon(1-\varepsilon)}e^{-i\theta'}\ketbra{\phi}{0}+\varepsilon\ketbra{\phi}\right)^{\otimes k}.
\end{split}
\end{equation}
By linearity, we have
\begin{equation}
    \E_{\substack{\phi,\psi\sim\mathbb{C}^d\\\theta,\theta'\sim[0,2\pi]}}\Tr(M(\Theta_{AB}-\Psi_{AB}))=\E_{\phi,\psi\sim\mathbb{C}^d}\Tr(M\left(\E_{\theta,\theta'\sim[0,2\pi]}\Theta_{AB}-\E_{\theta,\theta'\sim[0,2\pi]}\Psi_{AB}\right)).
\end{equation}
The random phases $\theta,\theta'$ are always independent, so the integral can be calculated separately for Alice and Bob. First consider Alice's system. The integral over the random phase $\theta$ gives
\begin{equation}\label{eq:lowerboundrandomphase}
\begin{split}
    &\E_{\theta\sim[0,2\pi]}\left((1-\varepsilon)\ketbra{0}+\sqrt{\varepsilon(1-\varepsilon)}e^{i\theta}\ketbra{0}{\phi}+\sqrt{\varepsilon(1-\varepsilon)}e^{-i\theta}\ketbra{\phi}{0}+\varepsilon\ketbra{\phi}\right)^{\otimes k}\\
    &=\sum_{t=0}^k\varepsilon^t(1-\varepsilon)^{k-t}\binom{k}{t}^2\psym{d}{k}\left(\ketbra{0}^{\otimes k-t}\otimes \ketbra{\phi}^{\otimes t}\right)\psym{d}{k}\\
    &=\sum_{t=0}^k\varepsilon^t(1-\varepsilon)^{k-t}\binom{k}{t}\ketbra{\phi_t},
\end{split}
\end{equation}
where the state $\ket{\phi_t}$ is defined as
\begin{equation}
    \ket{\phi_t}:=\frac{1}{\sqrt{\binom{k}{t}}}\sum_{S\subseteq[k],|S|=t}\ket{0}_{\bar{S}}\ket{\phi}_S,
\end{equation}
a uniform superposition of $\phi$ stored in $t$ of the $k$ registers. We define the state $\ket{\psi_t}$ similarly. Here the second line of Eq.~\eqref{eq:lowerboundrandomphase} is proven in the following lemma.

\begin{lemma}\label{lemma:lowerboundrandomphase}
For any $\varepsilon\in(0,1)$ and pure state $\ket{\phi}$, we have
\begin{equation}
    \begin{split}
        &\E_{\theta\sim[0,2\pi]}\left((1-\varepsilon)\ketbra{0}+\sqrt{\varepsilon(1-\varepsilon)}e^{i\theta}\ketbra{0}{\phi}+\sqrt{\varepsilon(1-\varepsilon)}e^{-i\theta}\ketbra{\phi}{0}+\varepsilon\ketbra{\phi}\right)^{\otimes k}\\
    &=\sum_{t=0}^k\varepsilon^t(1-\varepsilon)^{k-t}\binom{k}{t}^2\psym{d}{k}\left(\ketbra{0}^{\otimes k-t}\otimes \ketbra{\phi}^{\otimes t}\right)\psym{d}{k}.
    \end{split}
\end{equation}
\end{lemma}
\begin{proof}
First note that after expanding the tensor product, in each term the number of $\sqrt{\varepsilon(1-\varepsilon)}e^{i\theta}\ketbra{0}{\phi}$ must be equal to the number of $\sqrt{\varepsilon(1-\varepsilon)}e^{-i\theta}\ketbra{\phi}{0}$, because otherwise the term will be 0 after performing integral over the random phase $\theta$. This guarantees that in each term the number of 0s ($\phi$s) on the ket side equals the number of 0s ($\phi$s) on the bra side. Next, consider a term of the form
\begin{equation}\label{eq:lowerboundpermutationterm}
    \ket{\sigma(\phi,\phi,\dots,\phi,0,0,\dots,0)}\bra{\tau(\phi,\phi,\dots,\phi,0,0,\dots,0)}
\end{equation}
which has $k-t$ 0s and $t$ $\phi$s, and $\sigma,\tau$ are arbitrary permutations. Each term of this form has a one-to-one correspondence with the terms in the expansion of the tensor product. Moreover, regardless of $\sigma$ and $\tau$ the coefficient of this term is always
\begin{equation}
    \varepsilon^t(1-\varepsilon)^{k-t}.
\end{equation}
Finally, using the characterization of $\psym{d}{k}=\frac{1}{k!}\sum_{\pi\in S_k}P_d(\pi)$ (Lemma~\ref{lemma:symprojector}) it is easy to see that the sum of unique terms of the form in Eq.~\eqref{eq:lowerboundpermutationterm} equals
\begin{equation}
\begin{split}
    \binom{k}{t}^2\Pi_{\mathrm{sym}}^{d,k}\left(\ketbra{0}^{\otimes k-t}\otimes \ketbra{\phi}^{\otimes t}\right)\Pi_{\mathrm{sym}}^{d,k},
\end{split}
\end{equation}
which concludes the proof.
\end{proof}

The state after averaging over random phase can be understood as follows: first sample a random $t$ according to binomial distribution $B(k,\varepsilon)$, then the state is $t$ copies of $\phi$ uniformly distributed (in superposition) in $k$ registers. This gives the same average state as in Problem~\ref{prob:lowerbounddistinguish}, because 
\begin{equation}
    \begin{split}
        &\E_{\theta,\theta'\sim[0,2\pi]}\Theta_{AB}=\E_{t,t'\sim B(k,\varepsilon)}\ketbra{\phi_t}\otimes \ketbra{\psi_{t'}},\\
        &\E_{\theta,\theta'\sim[0,2\pi]}\Psi_{AB}=\E_{t,t'\sim B(k,\varepsilon)}\ketbra{\phi_t}\otimes \ketbra{\phi_{t'}}.
    \end{split}
\end{equation}

Effectively we have reduced Problem~\ref{prob:lowerbounddistinguish} to the following decision problem.

\begin{problem}
\label{prob:lowerbounddistinguish2}
Suppose Alice and Bob are each given one copy of a pure state in $(\mathbb{C}^{d+1})^{\otimes k}$. They are promised that one of the following two cases hold:
\begin{enumerate}
    \item Alice has $\ket{\phi_t}$ and Bob has $\ket{\phi_{t'}}$, where $\ket{\phi}$ is a uniformly random state in  $\vspan\{\ket{1},\ket{2},\dots,\ket{d}\}$, and $t,t'$ are independent random variables with binomial distribution $B(k,\varepsilon)$.
    \item Alice has $\ket{\phi_t}$ and Bob has $\ket{\psi_{t'}}$, where $\ket{\phi},\ket{\psi}$ are independent random states in  $\vspan\{\ket{1},\ket{2},\dots,\ket{d}\}$, and $t,t'$ are independent random variables with binomial distribution $B(k,\varepsilon)$.
\end{enumerate}
Their goal is to decide which case they are in with success probability at least $2/3$, using an interactive protocol that involves local quantum operations and classical communication.
\end{problem}

We have shown that Alice and Bob's average state is the same for Problem~\ref{prob:lowerbounddistinguish} and Problem~\ref{prob:lowerbounddistinguish2}. Therefore any measurement strategy will have the same success probability for Problem~\ref{prob:lowerbounddistinguish} and Problem~\ref{prob:lowerbounddistinguish2}, and any lower bound for Problem~\ref{prob:lowerbounddistinguish2} will also hold for Problem~\ref{prob:lowerbounddistinguish}. In the following we prove a $\Omega(\sqrt{d}/\varepsilon)$ lower bound for Problem~\ref{prob:lowerbounddistinguish2} via a reduction to DIPE.

The idea is to simulate the input to Problem~\ref{prob:lowerbounddistinguish2} using the input to DIPE. Suppose Alice and Bob are trying to solve DIPE with $m=100k\varepsilon+100$ copies, where in case 1 Alice and Bob both have $\ket{\phi}^{\otimes m}$, and in case 2 Alice has $\ket{\phi}^{\otimes m}$ and Bob has $\ket{\psi}^{\otimes m}$, where $\phi,\psi\sim\mathbb{C}^d$. Alice samples a random variable $t\sim B(k,\varepsilon)$, which satisfies $t\leq m$ with high probability. Alice throws away $m-t$ copies of her state, then applies a unitary that maps $\ket{\phi}^{\otimes t}\ket{0}^{\otimes k-t}$ to $\ket{\phi_t}$ as constructed in Lemma \ref{lemma:prepareuniformsuperposition} below. Bob also performs the same procedure. Then their resulting states will be close to the input states of Problem~\ref{prob:lowerbounddistinguish2}, where the two cases in DIPE also correspond to the two cases in Problem~\ref{prob:lowerbounddistinguish2}. Therefore they can decide DIPE with high probability using the algorithm for Problem~\ref{prob:lowerbounddistinguish2}.

In the above argument we have used the following lemma.
\begin{lemma}\label{lemma:prepareuniformsuperposition}
For any $t\in \{0,\dots,k\}$, there exists a unitary operator $U_t$ acting on $(\mathbb{C}^{d+1})^{\otimes k}$ that satisfies
\begin{equation}
    U_t \ket{\phi}^{\otimes t}\ket{0}^{\otimes k-t}=\ket{\phi_t}
\end{equation}
for any state $\ket{\phi}\in\vspan\{\ket{1},\dots,\ket{d}\}$.
\end{lemma}
\begin{proof}
It suffices to prove that the map $\ket{\phi}^{\otimes t}\ket{0}^{\otimes k-t}\mapsto\ket{\phi_t}$ preserves inner product. This follows from
\begin{equation}
    \begin{split}
        \braket{\phi_t'}{\phi_t}&=\frac{1}{\binom{k}{t}}\left(\sum_{S\subseteq[k],|S|=t}\bra{0}_{\bar{S}}\bra{\phi'}_S\right)\left(\sum_{T\subseteq[k],|T|=t}\ket{0}_{\bar{T}}\ket{\phi}_T\right)\\
        &=\frac{1}{\binom{k}{t}}\sum_{S\subseteq[k],|S|=t}\bra{0}_{\bar{S}}\bra{\phi'}_S\ket{0}_{\bar{S}}\ket{\phi}_S\\
        &=\braket{\phi'}{\phi}^t\\
        &=\bra{\phi'}^{\otimes t}\bra{0}^{\otimes k-t}\ket{\phi}^{\otimes t}\ket{0}^{\otimes k-t}
    \end{split}
\end{equation}
for any $\ket{\phi},\ket{\phi'}\in\vspan\{\ket{1},\dots,\ket{d}\}$. 

In fact a circuit for implementing $U_t$ is easy to describe. On input $\ket{\phi}^{\otimes t}\ket{0}^{\otimes k-t}$, we can prepare $\frac{1}{\sqrt{\binom{k}{t}}}\sum_{S}\ket{\phi}^{\otimes t}\ket{0}^{\otimes k-t}\ket{S}$, then use $S$ to perform control swap operations on the registers, which gives $\frac{1}{\sqrt{\binom{k}{t}}}\sum_{S}\ket{\phi}_S\ket{0}_{\bar{S}}\ket{S}$. Now, since $\ket{\phi}$ is supported on $\vspan\{\ket{1},\dots,\ket{d}\}$ and always orthogonal to $\ket{0}$, we can reversibly compute $S$ from $\ket{\phi}_S\ket{0}_{\bar{S}}$. This means that we can uncompute the $\ket{S}$ register from $\frac{1}{\sqrt{\binom{k}{t}}}\sum_{S}\ket{\phi}_S\ket{0}_{\bar{S}}\ket{S}$, resulting in $\ket{\phi_t}$.
\end{proof}

In more details, for any $\varepsilon>0$, suppose there is an algorithm $\mc A$ that can solve Problem~\ref{prob:lowerbounddistinguish2} using $k$ copies with success probability at least 0.99. Let $m=100k\varepsilon+100$. Consider the following algorithm for DIPE.\\

\noindent\textbf{Algorithm $\mc B$ for DIPE:}
\begin{enumerate}
    \item Alice has input state $\ket{\phi}^{\otimes m}$. She samples $t\sim B(k,\varepsilon)$. If $t>m$, she throws away all of her input states and prepares the state $\ket{A}=\ket{0}^{\otimes k}$. If $t\leq m$ she throws away $m-t$ copies of her input states and prepares $\ket{A}=\ket{\phi_t}$ using the procedure described in Lemma~\ref{lemma:prepareuniformsuperposition}.
    \item Bob samples $t'\sim B(k,\varepsilon)$ and prepares the state $\ket{B}$ similarly.
    \item They run $\mc A$ on $\ketbra{A}\otimes \ketbra{B}$, and return according to $\mc A$.
\end{enumerate}

First we bound the probability $\Pr[t>m]$. Let $X_1,\dots,X_k$ be i.i.d. 0/1 random variables with mean $\varepsilon$. Chernoff bound gives
\begin{equation}
    \Pr[\frac{X_1+\cdots X_k}{k}\geq \varepsilon+\delta]\leq \exp\left(-\frac{k\delta^2}{2(\varepsilon+\delta)}\right).
\end{equation}
Let $\delta=99\varepsilon+\frac{100}{k}>\varepsilon$, then
\begin{equation}
\begin{split}
    \Pr[t>m]&\leq \Pr[X_1+\cdots X_k\geq 100k\varepsilon+100]\\
    &=\Pr[\frac{X_1+\cdots X_k}{k}\geq\varepsilon+ \delta]\\
    &\leq \exp\left(-\frac{k\delta^2}{2(\varepsilon+\delta)}\right)\\
    &\leq \exp\left(-\frac{k\delta}{4}\right)\\
    &\leq\exp\left(-25\right)<0.01.
\end{split}
\end{equation}
Here, the last line uses $k\delta \geq 100$. To prove that algorithm $\mc B$ decides DIPE with high probability, it suffices to show that in both case 1 and case 2 the average state that Alice and Bob prepare is close to the average state in Problem~\ref{prob:lowerbounddistinguish2}. This follows from the above bound on failure probability. In case 1 the average state that Alice and Bob prepare is
\begin{equation}
    \E_{\phi\sim\mathbb{C}^d}\left(\sum_{t=0}^m \binom{k}{t}\varepsilon^t (1-\varepsilon)^{k-t}\ketbra{\phi_t}+\Pr[t>m]\ketbra{0}^{\otimes k}\right)^{\otimes 2}.
\end{equation}
It's easy to see that the trace distance between this state and the average state for Problem~\ref{prob:lowerbounddistinguish2} can be upper bounded by
\begin{equation}
\begin{split}
    &\left\|\E_{\phi\sim\mathbb{C}^d}\left(\sum_{t=0}^m \binom{k}{t}\varepsilon^t (1-\varepsilon)^{k-t}\ketbra{\phi_t}+\Pr[t>m]\ketbra{0}^{\otimes k}\right)^{\otimes 2}-\E_{\phi\sim\mathbb{C}^d}\E_{t,t'\sim B(k,\varepsilon)}\ketbra{\phi_t}\otimes \ketbra{\phi_{t'}}\right\|_1\\
    &\leq 4\Pr[t>m]<0.04,
\end{split}
\end{equation}
and the same bound also holds for case 2. Therefore, since $\mc A$ decides Problem~\ref{prob:lowerbounddistinguish2} with probability at least 0.99, algorithm $\mc B$ decides DIPE with high probability. Then the $\Omega(\sqrt{d})$ lower bound for DIPE implies that
\begin{equation}
    m=100k\varepsilon+100=\Omega(\sqrt{d}),
\end{equation}
which gives $k=\Omega(\sqrt{d}/\varepsilon)$.

\subsection{Optimality of the SWAP test}
\label{sec:swaptest}
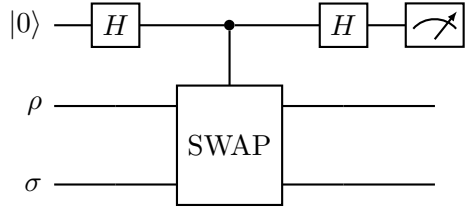
\begin{figure}[t]
    \centering
    \begin{quantikz}
\lstick{$\ket{0}$} & \gate{H} & \ctrl{1}  &  \gate{H} & \meter{} \\
\lstick{$\rho$} &\qw & \gate[wires=2]{\mathrm{SWAP}} & \qw & \qw \\
\lstick{$\sigma$} &\qw & &\qw & \qw
\end{quantikz}
    \caption{The SWAP test for estimating the inner product. The probability of seeing measurement outcome 0 at the output qubit equals $\frac{1+\Tr(\rho\sigma)}{2}$.}
    \label{fig:swaptest}
\end{figure}

Recall that in the non-distributed setting where arbitrary quantum operations on $\rho^{\otimes k}\otimes \sigma^{\otimes k}$ are allowed, we can estimate $f=\Tr(\rho \sigma)$ using $k=O(1/\varepsilon^2)$ copies using SWAP test. The SWAP test refers to the well-known circuit for estimating inner product shown in Fig.~\ref{fig:swaptest}. The probability of seeing 0 at the output qubit equals $\frac{1+\Tr(\rho\sigma)}{2}$. In experiment, given $\rho^{\otimes k}$ and $\sigma^{\otimes k}$, we repeat the SWAP test $k$ times and use the $k$ bits output $\{x_1,\dots,x_k\}$ to estimate $f=\Tr(\rho\sigma)$ as $\tilde{f}=1-2\cdot \frac{x_1+\cdots +x_k}{k}$, which is an unbiased estimator for $f$. The variance is given by $\Var(\tilde{f})=\frac{1-f^2}{k}$. Therefore $k=O(1/\varepsilon^2)$ copies is sufficient to estimate $f$ within $\varepsilon$ additive error with probability at least $2/3$.

As the SWAP test only uses two-copy measurements, it is natural to ask if we can do better with general multi-copy measurements. An improved algorithm for estimating the inner product with multi-copy measurements was given by~\cite{Badescu2019quantum}, which we refer to as generalized SWAP test. A related work by \cite{Fanizza2020beyond} addresses the optimal precision of inner product estimation with a given number of copies.

\begin{lemma}[Generalized SWAP test \cite{Badescu2019quantum}]\label{lemma:generalizedswaptest}
There exists an algorithm which takes input $\rho^{\otimes k}\otimes\sigma^{\otimes k}$ outputs an unbiased estimator for $f=\Tr(\rho\sigma)$ using multi-copy measurements that has variance
\begin{equation}
    \frac{1}{k^2}+\frac{k-1}{k^2}\left(\Tr(\rho^2\sigma)+\Tr(\rho\sigma^2)\right)-\frac{2k-1}{k^2}f^2.
\end{equation}
This algorithm is optimal in the sense of having optimal variance among all unbiased estimators.
\end{lemma}

In the worst case the optimal variance is $\Omega(1/k)$. As $\Tr(\rho^2\sigma)$ and $\Tr(\rho\sigma^2)$ can be upper bounded by $f$, this implies that 
\begin{equation}
    k=O\left(\frac{1}{\varepsilon}+\frac{f}{\varepsilon^2}\right)
\end{equation}
copies is sufficient to estimate $f$ within $\varepsilon$ additive error with probability at least $2/3$.

Both the standard and the generalized SWAP test requires $\Omega(1/\varepsilon^2)$ copies in general. In the following we show that this is optimal by proving a $\Omega(1/\varepsilon^2)$ sample complexity lower bound for estimating inner product.

\begin{lemma}\label{lemma:swaptestlowerbound}
Suppose there is an algorithm that on input $\rho^{\otimes k}\otimes\sigma^{\otimes k}$, where $\rho,\sigma$ act on $\mathbb{C}^{d}$, outputs an estimate of $\Tr(\rho\sigma)$ within $\varepsilon$ additive error with success probability at least $2/3$. Then $k=\Omega\left(1/\varepsilon^2\right)$.
\end{lemma}
\begin{proof}
Define the states $\ket{\psi_0},\ket{\psi_1}$ as
\begin{equation}
\begin{split}
    \ket{\psi_0}&=\sqrt{\frac{1}{2}-\varepsilon}\ket{0}+\sqrt{\frac{1}{2}+\varepsilon}\ket{1},\\
    \ket{\psi_1}&=\sqrt{\frac{1}{2}+\varepsilon}\ket{0}+\sqrt{\frac{1}{2}-\varepsilon}\ket{1}.
\end{split}
\end{equation}
Consider the task of distinguishing between the two cases (1) $\psi_0^{\otimes k}\otimes \ketbra{0}^{\otimes k}$, and (2) $\psi_1^{\otimes k}\otimes \ketbra{0}^{\otimes k}$. In case 1 the inner product is given by $\left|\braket{0}{\psi_0}\right|^2=\frac{1}{2}-\varepsilon$, and in case 2 the inner product is given by $\left|\braket{0}{\psi_1}\right|^2=\frac{1}{2}+\varepsilon$. Therefore using the algorithm for inner product estimation we can distinguish between the two cases with high probability. This implies that the fidelity between the two cases is at most $1-\Omega(1)$.

On the other hand, the fidelity between the two cases is given by
\begin{equation}
    1-\Omega(1)\geq F(\psi_0^{\otimes k}\otimes \ketbra{0}^{\otimes k},\psi_1^{\otimes k}\otimes \ketbra{0}^{\otimes k})=F(\psi_0,\psi_1)^k=(1-4\varepsilon^2)^k\geq 1-4k\varepsilon^2,
\end{equation}
which implies that $k=\Omega(1/\varepsilon^2)$.
\end{proof}

Clearly, this lower bound also applies to our distributed setting as it is a special case of non-distributed setting. Combining this with our lower bound in Theorem~\ref{thm:epsestimationlowerbound}, we have the following stronger lower bound for $\ipe$.

\begin{corollary}\label{cor:lowerbound}
Suppose Alice has input $\rho^{\otimes k}$ and Bob has input $\sigma^{\otimes k}$, for arbitrary unknown mixed states $\rho,\sigma$ acting on $\mathbb{C}^{d}$. Then 
\begin{equation}
    k=\Omega\left(\max\left\{\frac{1}{\varepsilon^2},\frac{\sqrt{d}}{\varepsilon}\right\}\right)
\end{equation}
copies is necessary for them to estimate $\Tr(\rho\sigma)$ within additive error $\varepsilon$ with success probability $2/3$, even when they are allowed arbitrary multi-copy measurements and interactive protocols.
\end{corollary}

\section{Inner product estimation with single-copy measurements}
\label{sec:singlecopy}
Next we show how to solve quantum inner product estimation using single-copy measurements. The algorithm only uses non-adaptive measurements and simultaneous message passing for communication (Fig.~\ref{fig:algtemplate}b). The high-level idea is to reduce quantum inner product estimation to classical inner product estimation via performing a random projection on the quantum states. The general idea of reducing quantum property estimation to classical property estimation via random projection has been used in other tasks such as shadow tomography~\cite{Huang2020predicting} and state certification~\cite{Bubeck2020entanglement}. Here we show this technique is also effective for learning multiple quantum systems.

\subsection{Classical collision estimator}
We start by reviewing the classical collision estimator for estimating the inner product of two discrete probability distributions. The idea of collision estimator is to count the pairwise collisions between two sets of samples from two discrete distributions, which has been widely used in property testing of probability distributions~\cite{Diakonikolas2019collision,Canonne2020survey}.

Consider two unknown distributions $p,q$ supported on $\{0,\dots,d-1\}$. We are given i.i.d. samples $x_1,\dots,x_m\sim p$ and $y_1,\dots,y_m\sim q$, and the goal is to estimate the inner product $g:=\sum_{b=0}^{d-1} p_b q_b$.

\begin{definition}[Collision estimator]
Given samples $x_1,\dots,x_m\sim p$ and $y_1,\dots,y_m\sim q$ from two discrete distributions $p,q$, the collision estimator is defined as
\begin{equation}\label{eq:classicalestimator}
    \tilde{g}=\frac{1}{m^2}\sum_{j,k=1}^m 1[x_j=y_k]
\end{equation}
where $1[X]=1$ when $X$ is true and 0 otherwise.
\end{definition}

It is easy to see that the collision estimator is an unbiased estimator for the inner product. We can also calculate its variance as follows.

\begin{lemma}\label{lemma:collisionvariance}
The variance of the collision estimator is upper bounded by
\begin{equation}
        \Var(\tilde{g})\leq \frac{g}{m^2}+\frac{1}{m}\left(\sum_{b=0}^{d-1} p_b q_b^2+\sum_{b=0}^{d-1} p_b^2 q_b\right).
\end{equation}
\end{lemma}
\begin{proof}
The variance is given by
\begin{equation}
    \begin{split}
        \Var(\tilde{g})&=\E \tilde{g}^2-g^2\\
        &=\frac{1}{m^4}\sum_{j,k,l,s=1}^m\E 1[x_j=y_k]\cdot 1[x_l=y_s]-g^2\\
        &=\frac{1}{m^4}\left(m^2 g+m^2(m-1)^2 g^2+m^2(m-1)\sum_{b=0}^{d-1} p_b q_b^2+m^2(m-1)\sum_{b=0}^{d-1} p_b^2 q_b\right)-g^2\\
        &\leq \frac{g}{m^2}+\frac{1}{m}\left(\sum_{b=0}^{d-1} p_b q_b^2+\sum_{b=0}^{d-1} p_b^2 q_b\right).
    \end{split}
\end{equation}
In the third line, the four terms in the bracket come from 4 different cases in the sum: $j=l\wedge k=s$, $j\neq l\wedge k\neq s$, $j=l\wedge k\neq s$, $j\neq l\wedge k= s$, respectively.
\end{proof}

The variance can be further upper bounded by $\Var(\tilde{g})\leq O(g/m)$, which implies the well-known fact that $O(g/\varepsilon^2)$ samples suffices to estimate $g$ within $\varepsilon$ additive error with success probability at least $2/3$. However in our case it is crucial to keep the original expression in Lemma~\ref{lemma:collisionvariance}, as we will see later that $p$ and $q$ will be close to uniform, which significantly reduces the variance compared to the worst case bound $O(g/m)$.

\begin{figure}[t]
  \begin{algorithm}[H]
    \caption{Quantum inner product estimation with single-copy measurement}\label{alg:innerproduct}
    \raggedright\textbf{Input:} number of basis settings $N$, number of measurements for each basis $m$, $k=Nm$ copies of unknown mixed states $\rho,\sigma$ acting on $\mathbb{C}^{d}$\\
    \textbf{Output:} an estimate of $f:=\Tr(\rho\sigma)$
    \begin{algorithmic}[1]
    \For {$i=1\dots N$}
        \State sample a random unitary matrix $U_i\sim\mathbb{U}(d)$
        \State measure $m$ copies of $\rho$ in the basis $\{U_i^\dag\ketbra{b}U_i\}_{b=0}^{d-1}$ and obtain $X=\{x_1,\dots,x_m\}$
        \State measure $m$ copies of $\sigma$ in the basis $\{U_i^\dag\ketbra{b}U_i\}_{b=0}^{d-1}$ and obtain $Y=\{y_1,\dots,y_m\}$
        \State compute the collision estimator~\eqref{eq:classicalestimator} using $X$ and $Y$, denote by $\tilde{g}_i$
        \State let $w_i=(d+1)\tilde{g}_i -1$
    \EndFor
    \State\textbf{Return} $w:=\frac{1}{N}\sum_{i=1}^N w_i$
    \end{algorithmic}
    \end{algorithm}
\end{figure}

\subsection{Analysis of the estimation algorithm}

Next we provide an analysis of the estimation algorithm (Algorithm~\ref{alg:innerproduct}) which uses the template in Fig.~\ref{fig:algtemplate}b. A variant of Algorithm~\ref{alg:innerproduct} was proposed by \cite{Elben2020cross} which also provided a numerical analysis of the variance. The only difference is that in Algorithm~\ref{alg:innerproduct} we use the collision estimator to estimate the classical inner product, while \cite{Elben2020cross} uses the inner product between the empirical distributions. We remark that the two estimators should have similar performance, and here we use the collision estimator because it is easier to analyze.

The main idea of Algorithm~\ref{alg:innerproduct} is to rotate $\rho,\sigma$ by a random unitary matrix $U$, then estimate the inner product of the distributions $p(U),q(U)$, where
\begin{equation}
    p_b(U)=\expval{U\rho U^\dag}{b},\,\,\,\,q_b(U)=\expval{U\sigma U^\dag}{b},
\end{equation}
which are the discrete probability distributions that correspond to the measurement outcome of $U\rho U^\dag$ and $U\sigma U^\dag$ in the computational basis. Their inner product is defined as
\begin{equation}
    g(U)=\sum_{b=0}^{d-1}\expval{U\rho U^\dag}{b}\expval{U\sigma U^\dag}{b}.
\end{equation}

Let $\tilde{g}(U,S)$ denote the collision estimator (line 5 of Algorithm~\ref{alg:innerproduct}) in a single iteration, with unitary $U$ and samples $S=\{X,Y\}$. Its mean is given by
\begin{equation}
\begin{split}
    \E_{U,S}\tilde{g}(U,S)&=\E_U g(U)\\
    &=\sum_{b=0}^{d-1}\E_U\expval{U\rho U^\dag}{b}\expval{U\sigma U^\dag}{b}\\
    &=d\E_{\psi\sim\mathbb{C}^d}\expval{\rho}{\psi}\expval{\sigma}{\psi}\\
    &=\frac{1}{d+1}\Tr(\left(I\otimes I+\swapop\right)\rho\otimes \sigma)\\
    &=\frac{1+f}{d+1}.
\end{split}
\end{equation}
Here the fourth line follows from the characterization of $\psym{d}{2}$ (Lemma~\ref{lemma:symprojector}, also see Lemma~\ref{lemma:randomstatemoments}). This implies the following.
\begin{lemma}\label{lemma:singlecopyunbiased}
Let $w$ be the estimator returned by Algorithm~\ref{alg:innerproduct}. Then $w$ is an unbiased estimator for $f=\Tr(\rho\sigma)$.
\end{lemma}

To calculate the variance of $w$ it suffices to calculate the variance of $\tilde{g}(U,S)$. This is because each iteration $i=1,\dots,N$ is independent and
\begin{equation}
    \Var(w)=\frac{1}{N^2}\sum_{i=1}^N \Var(w_i)=\frac{1}{N}(d+1)^2\Var(\tilde{g}(U,S)).
\end{equation}

The total variance $\Var(\tilde{g}(U,S))$ involves two parts: the variance introduced by the random unitary $U$ and the variance introduced by the intrinsic randomness of quantum measurement. The law of total variance gives
\begin{equation}\label{eq:totalvariance}
    \Var(\tilde{g}(U,S))=\Var_U(\E[\tilde{g}(U,S)|U])+\E_U[\Var(\tilde{g}(U,S)|U)],
\end{equation}
which precisely corresponds to the two parts described above. These terms can be upper bounded as follows.

\begin{lemma}\label{lemma:singlecopyvariance}
The variance of $\tilde{g}(U,S)$ with respect to $U$ and $S$ can be upper bounded by
\begin{equation}
    \Var_U(\E[\tilde{g}(U,S)|U])=O\left(\frac{1}{d^3}\right),\,\,\,\,\E_U[\Var(\tilde{g}(U,S)|U)]=O\left(\frac{1}{m^2 d}+\frac{1}{m d^2}\right).
\end{equation}
Therefore the total variance is
\begin{equation}
    \Var(\tilde{g}(U,S))=O\left(\frac{1}{d^3}+\frac{1}{m^2 d}+\frac{1}{m d^2}\right).
\end{equation}
\end{lemma}

The proof is given in Appendix~\ref{app:singlecopyvariance}, where we also show that the above bounds are tight. In particular, when $\rho,\sigma$ are pure states the exact form of the variance is given by
\begin{equation}\label{eq:singlecopyvarianceconstantfactor}
\begin{split}
    \Var(\tilde{g}(U,S))&=\frac{d^2 (1+f)^2-d (6-f) f+d+2 (1-f)^2}{d (d+1)^2 (d+2) (d+3)}\\
    &+
    \frac{1+f}{(d+1)m^2}+\frac{m-1}{m^2}\frac{4+8f}{(d+1)(d+2)}-\frac{2m-1}{m^2}\frac{d^2 (1+f)^2+5 d (1+f)^2+2 (1-f)^2}{d (d+1) (d+2) (d+3)}\\
    &=\Omega(1)\cdot \left(\frac{1}{d^3}+\frac{1}{m^2 d}+\frac{1}{m d^2}\right),
\end{split}
\end{equation}
where the first line corresponds to $\Var_U(\E[\tilde{g}(U,S)|U])$, the second line corresponds to $\E_U[\Var(\tilde{g}(U,S)|U)]$, $f=\Tr(\rho\sigma)$ and the constant in $\Omega(1)$ does not depend on $f$.

Combining Lemma~\ref{lemma:singlecopyunbiased} and Lemma~\ref{lemma:singlecopyvariance} we have the following characterization of the performance of Algorithm~\ref{alg:innerproduct}.

\begin{theorem}\label{thm:singlecopyalg}
Let $w$ be the estimator returned by Algorithm~\ref{alg:innerproduct}, which uses $N$ random unitary bases and $m$ copies for each basis. Then we have $\E w=f$ and
\begin{equation}\label{eq:singlecopyvariance}
    \Var(w)=\frac{1}{N}O\left(\frac{1}{d}+\frac{d}{m^2}+\frac{1}{m}\right).
\end{equation}
The sample complexity is given by $k=N m$. When $\varepsilon=\Omega(1)$, it suffices to choose $N=1$ and $m=O(\sqrt{d}/\varepsilon)$ to achieve $\varepsilon$ additive error with probability at least $2/3$. In general, for any $\varepsilon>0$, Algorithm~\ref{alg:innerproduct} returns an estimate of $f$ within $\varepsilon$ additive error with probability at least $2/3$, provided that
\begin{equation}
    k\geq C\cdot \max\left\{\frac{1}{\varepsilon^2},\frac{\sqrt{d}}{\varepsilon}\right\}
\end{equation}
for some constant $C>0$ that does not depend on $f$.
\end{theorem}
\begin{proof}
For any $\varepsilon\in(0,1)$, from the three terms in Eq.~\eqref{eq:singlecopyvariance} it is necessary and sufficient to have
\begin{equation}
    N\geq\max\left\{1,\frac{1}{d \varepsilon^2}\right\},\,\,\,\,Nm^2\geq \frac{d}{\varepsilon^2},\,\,\,\,Nm\geq\frac{1}{\varepsilon^2}.
\end{equation}
Here we ignored constants. The third condition implies that $k\geq \frac{1}{\varepsilon^2}$. The first two conditions imply that
\begin{equation}
    k\geq N\cdot \frac{1}{\sqrt{N}}\cdot \frac{\sqrt{d}}{\varepsilon}\geq\frac{\sqrt{d}}{\varepsilon}\cdot \max\left\{1,\frac{1}{\sqrt{d} \varepsilon}\right\}=\max\left\{\frac{1}{\varepsilon^2},\frac{\sqrt{d}}{\varepsilon}\right\}.
\end{equation}
This lower bound is also achievable as we can choose $N=\max\left\{1,\frac{1}{d \varepsilon^2}\right\}$ and $m=\frac{\sqrt{d}}{\varepsilon\sqrt{N}}$.
\end{proof}

Combining our lower bound in Corollary~\ref{cor:lowerbound} with Theorem~\ref{thm:singlecopyalg} we have the following.
\begin{theorem}[Main result, restatement]
For any $\varepsilon\in(0,1)$, the sample complexity of $\ipe$ is
\begin{equation}
    k=\Theta\left(\max\left\{\frac{1}{\varepsilon^2},\frac{\sqrt{d}}{\varepsilon}\right\}\right)
\end{equation}
across all measurement and communication settings.
\end{theorem}

Therefore our results completely settle the sample complexity of distributed quantum inner product estimation in the general setting without any prior knowledge on $f$.

Interestingly, if we are promised an upper bound on $f$, then for pure states our algorithm with multi-copy measurements (Algorithm~\ref{alg:innerproductmulti}) has better sample complexity. For this, note that in Eq.~\eqref{eq:multicopysamplecomplexity} the first two terms in the RHS have respective factors of $f$ and $\sqrt{f}$. If $f$ is promised to be small, then the number of copies with Algorithm~\ref{alg:innerproductmulti} clearly improve in comparison to the worst case ($f=1$). In contrast, our algorithm with single-copy measurements (Algorithm~\ref{alg:innerproduct}) does not have this feature due to the constant prefactor (as shown in Eq.~\eqref{eq:singlecopyvarianceconstantfactor}) that does not depend on $f$.

\section*{Acknowledgements}
We thank Srinivasan Arunachalam, Soonwon Choi, Aram Harrow, Chinmay Nirkhe, John Wright, Guanyu Zhu and especially Umesh Vazirani for helpful comments and discussions. We would also like to thank the Simons Institute for the Theory of Computing where part of this work was done. This work was supported by NSF QLCI program through grant number OMA-2016245, DOE NQISRC QSA grant \#FP00010905, Vannevar Bush faculty fellowship N00014-17-1-3025, MURI Grant FA9550-18-1-0161. Y.L. was also supported by NSF award DMR-1747426.

\printbibliography

\appendix

\section{Proof of Lemma~\ref{lemma:varianceofmulticopyestimation}}
\label{app:varianceofmulticopyestimation}
Let $w=(\left|\braket{u}{v}\right|^2-A)/B$ where $A=\frac{d+2k}{(d+k)^2}$, $B=\frac{k^2}{(d+k)^2}$. The variance is given by
\begin{equation}
\begin{split}
    \Var(w)&=\E w^2-f^2\\
    &=\frac{1}{B^2}\E\left|\braket{u}{v}\right|^4-\frac{A^2}{B^2}-\frac{2A f}{B}-f^2\\
    &=\frac{(d+k)^4}{k^4}\E\left|\braket{u}{v}\right|^4-\frac{(d+2k)^2}{k^4}-\frac{2(d+2k)f}{k^2}-f^2.
\end{split}
\end{equation}
Therefore it suffices to calculate $\E\left|\braket{u}{v}\right|^4$. Recall that 
\begin{equation}
    \braket{u}{v}=\alpha\beta e^{i(\varphi-\theta)}\braket{\phi}{\psi}+\sqrt{1-\alpha^2}\beta e^{i\varphi}\braket{\phi'}{\psi}+\sqrt{1-\beta^2}\alpha e^{-i\theta}\braket{\phi}{\psi'}+\sqrt{(1-\alpha^2)(1-\beta^2)}\braket{\phi'}{\psi'},
\end{equation}
where $\varphi,\theta$ are random phases. The first step in calculating $\E\left|\braket{u}{v}\right|^4$ is to perform the integral with respect to the random phases. This can be achieved by the following lemma.

\begin{lemma}
\label{lem:complexlem}
Let $q,g,h,l$ be complex numbers. Then  
\begin{equation}\label{eq:complexlem}
\begin{split}
    &\E_{\theta,\phi\sim[0,2\pi]}\left|e^{i(\phi-\theta)}q+e^{i\phi}g+e^{-i\theta}h+l\right|^4\\
    &=|q|^4+|g|^4+|h|^4+|l|^4\\
    &+4\left(|g|^2|h|^2+|l|^2|g|^2+|l|^2|h|^2+|q|^2|g|^2+|q|^2|h|^2+|q|^2|l|^2+qlg^*h^*+ghq^*l^*\right).
\end{split}
\end{equation}
\end{lemma}
The proof is by brute force calculation. To use this lemma, we assign the complex numbers $q,g,h,l$ as
\begin{equation}
    \begin{split}
       q&=\alpha\beta \braket{\phi}{\psi}\\
        g&=\sqrt{1-\alpha^2}\beta \braket{\phi'}{\psi}\\
        h&=\sqrt{1-\beta^2}\alpha\braket{\phi}{\psi'}\\
        l&=\sqrt{(1-\alpha^2)(1-\beta^2)}\braket{\phi'}{\psi'}.
    \end{split}
\end{equation}
It remains to calculate each term in RHS of Eq.~\eqref{eq:complexlem}. The key tools to calculate these terms are the moments of the random state $\ketbra{\phi'}$. Recall that for an arbitrary state $\ket{\phi}\in\mathbb{C}^d$, the random state $\ket{\phi'}$ is defined as a uniformly random state in $\mathbb{C}^{d-1}_{\phi_\perp}$, which is the $d-1$ dimensional subspace of $\mathbb{C}^d$ that is perpendicular to $\ket{\phi}$. The moments of $\ketbra{\phi'}$ is given as follows.

\begin{lemma}
\label{lemma:symperpprojector}
Let $\ket{\psi}\in\mathbb{C}^d$ be an arbitrary vector, and let $\mathbb{C}^{d-1}_{\psi_\perp}$ be the $d-1$ dimensional subspace of $\mathbb{C}^d$ that is perpendicular to $\ket{\psi}$. Then
\begin{equation}
    \E_{\psi'\sim \mathbb{C}^{d-1}_{\psi_\perp}}\ketbra{\psi'}^{\otimes k}=\frac{1}{(d+k-2)\cdots (d-1)}\left(I-\ketbra{\psi}\right)^{\otimes k}\sum_{\pi\in S_k}P_d(\pi)\left(I-\ketbra{\psi}\right)^{\otimes k}.
\end{equation}
In particular, the $k=1,2$ moments are given as
\begin{equation}
    \begin{split}
        &\E_{\psi'\sim \mathbb{C}^{d-1}_{\psi_\perp}}\ketbra{\psi'}=\frac{I-\ketbra{\psi}}{d-1},\\
        &\E_{\psi'\sim \mathbb{C}^{d-1}_{\psi_\perp}}\ketbra{\psi'}^{\otimes 2}=\frac{1}{d(d-1)}\left(\left(I-\ketbra{\psi}\right)^{\otimes 2}+\left(I-\ketbra{\psi}\right)^{\otimes 2}\swapop\left(I-\ketbra{\psi}\right)^{\otimes 2}\right).
    \end{split}
\end{equation}
\end{lemma}
\begin{proof}
The proof follows from a simple observation. Note that LHS is proportional to the orthogonal projector onto $\vee^k \mathbb{C}^{d-1}_{\psi_\perp}$. Lemma~\ref{lemma:symprojector} implies that 
\begin{equation}
\begin{split}
    \E_{\psi'\sim \mathbb{C}^{d-1}_{\psi_\perp}}\ketbra{\psi'}^{\otimes k}&=\frac{1}{(d+k-2)\cdots (d-1)}\left(\text{sum of all permutations on }(\mathbb{C}^{d-1}_{\psi_\perp})^{\otimes k}\right)\\
    &=\frac{1}{(d+k-2)\cdots (d-1)}\left(I-\ketbra{\psi}\right)^{\otimes k}\sum_{\pi\in S_k}P_d(\pi)\left(I-\ketbra{\psi}\right)^{\otimes k},
\end{split}
\end{equation}
where the second line follows from the fact that each permutation operator on $(\mathbb{C}^{d-1}_{\psi_\perp})^{\otimes k}$ equals the same permutation on $(\mathbb{C}^{d})^{\otimes k}$ with the basis $\ketbra{\psi}$ being projected out.
\end{proof}

We define the following operator
\begin{equation}
    \swapop_\psi:=\left(I-\ketbra{\psi}\right)^{\otimes 2}\swapop\left(I-\ketbra{\psi}\right)^{\otimes 2}
\end{equation}
which is the $\swapop$ operator projected onto $(\mathbb{C}^{d-1}_{\psi_\perp})^{\otimes 2}$.

Using Lemma~\ref{lemma:symperpprojector} it is easy to calculate the following useful quantities.

\begin{lemma}
\begin{equation}
    \E_{\phi'\sim\mathbb{C}^{d-1}_{\phi_\perp}} \left|\braket{\phi'}{\psi}\right|^4=\frac{2(1-f)^2}{d(d-1)}.
\end{equation}
\end{lemma}
\begin{proof}
\begin{equation}
    \begin{split}
        \E \left|\braket{\phi'}{\psi}\right|^4&=\frac{1}{d(d-1)}\Tr(\left((I-\ketbra{\phi})^{\otimes 2}+\mathrm{SWAP}_{\phi}\right)\ketbra{\psi}^{\otimes 2})\\
        &=\frac{1}{d(d-1)}\left((1-f)^2+\Tr(\mathrm{SWAP}_{\phi}\ketbra{\psi}^{\otimes 2})\right)\\
        &=\frac{2(1-f)^2}{d(d-1)}.
    \end{split}
\end{equation}
\end{proof}

\begin{lemma}
\begin{equation}
    \E_{\substack{\phi'\sim\mathbb{C}^{d-1}_{\phi_\perp}\\\psi'\sim\mathbb{C}^{d-1}_{\psi_\perp}}}\left|\braket{\phi'}{\psi'}\right|^2 \left|\braket{\phi'}{\psi}\right|^2=\frac{(1-f) (d+2 f-2)}{d(d-1)^2}.
\end{equation}
\end{lemma}
\begin{proof}
\begin{equation}
    \begin{split}
        \E\left|\braket{\phi'}{\psi'}\right|^2 \left|\braket{\phi'}{\psi}\right|^2&=\Tr(\frac{(I-\ketbra{\phi})^{\otimes 2}+\mathrm{SWAP}_{\phi}}{d(d-1)}\cdot\frac{I-\ketbra{\psi}}{d-1}\otimes\ketbra{\psi})\\
        &=\frac{1}{d(d-1)^2}\left((d-2+f)(1-f)+\Tr(\mathrm{SWAP}_{\phi}\cdot I\otimes\ketbra{\psi})-(1-f)^2\right)\\
        &=\frac{1}{d(d-1)^2}\left((d+2f-3)(1-f)+1-f\right)\\
        &=\frac{(1-f) (d+2 f-2)}{d(d-1)^2}.
    \end{split}
\end{equation}
\end{proof}

\begin{lemma}
\begin{equation}
    \E_{\substack{\phi'\sim\mathbb{C}^{d-1}_{\phi_\perp}\\\psi'\sim\mathbb{C}^{d-1}_{\psi_\perp}}} \left|\braket{\phi'}{\psi'}\right|^4=\frac{2(d-2+f)^2+2(d-2+f^2)}{d^2(d-1)^2}.
\end{equation}
\end{lemma}
\begin{proof}
\begin{equation}
    \begin{split}
         \E \left|\braket{\phi'}{\psi'}\right|^4&=\frac{1}{d^2(d-1)^2}\Tr(\left((I-\ketbra{\phi})^{\otimes 2}+\mathrm{SWAP}_{\phi}\right)\left((I-\ketbra{\psi})^{\otimes 2}+\mathrm{SWAP}_{\psi}\right))\\
         &=\frac{(d-2+f)^2+2(d-2+f^2)+\Tr(\mathrm{SWAP}_{\phi}\cdot\mathrm{SWAP}_{\psi})}{d^2(d-1)^2},
    \end{split}
\end{equation}
where we have used $\Tr((I-\ketbra{\phi})^{\otimes 2}\cdot (I-\ketbra{\psi})^{\otimes 2})=(d-2+f)^2$, and
\begin{equation}
\begin{split}
    \Tr(\mathrm{SWAP}_{\phi}\cdot (I-\ketbra{\psi})^{\otimes 2})&=\Tr(\swapop (I-\ketbra{\phi})^{\otimes 2}(I-\ketbra{\psi})^{\otimes 2}(I-\ketbra{\phi})^{\otimes 2})\\
    &=\Tr((I-\ketbra{\phi})(I-\ketbra{\psi})(I-\ketbra{\phi})(I-\ketbra{\psi}))\\
    &=d-2+f^2.
\end{split}
\end{equation}
Finally,
\begin{equation}
\begin{split}
    \Tr(\mathrm{SWAP}_{\phi}\cdot\mathrm{SWAP}_{\psi})&=\left(\Tr((I-\ketbra{\psi})(I-\ketbra{\phi}))\right)^2\\
    &=(d-2+f)^2.
\end{split}
\end{equation}
\end{proof}

Having the above lemmas, and recall the fact that $\E\alpha^2=\E\beta^2=\frac{k+1}{d+k}$ and $\E \alpha^4=\E\beta^4=\frac{(k+2)(k+1)}{(d+k+1)(d+k)}$, we are now ready to calculate each term in RHS of Eq.~\eqref{eq:complexlem}.

Let's start with the last two terms. Using $\E\ketbra{\phi'}=\frac{I-\ketbra{\phi}}{d-1}$ and $\E\ketbra{\psi'}=\frac{I-\ketbra{\psi}}{d-1}$, we can show that
\begin{equation}
    \E \braket{\phi}{\psi}\braket{\psi}{\phi'}\braket{\phi'}{\psi'}\braket{\psi'}{\phi}=\frac{f^2-f}{(d-1)^2}.
\end{equation}
This gives
\begin{equation}
    \E qlg^*h^*=\E ghq^*l^*=\frac{(d-1)^2(k+1)^2}{(d+k)^2(d+k+1)^2}\cdot \frac{f^2-f}{(d-1)^2}=\frac{(k+1)^2}{(d+k)^2(d+k+1)^2}\cdot(f^2-f).
\end{equation}

Next, we have
\begin{equation}
    \begin{split}
        \E |q|^4&=\E \alpha^4\beta^4 f^2=\frac{(k+2)^2(k+1)^2}{(d+k+1)^2(d+k)^2} f^2\\
        \E |g|^4&=\E |h|^4\\
        &=\E(1-\alpha^2)^2\beta^4 \left|\braket{\phi'}{\psi}\right|^4\\
        &=\frac{(d-1) d (k+1) (k+2)}{(d+k)^2 (d+k+1)^2}\cdot (1-f)^2\cdot \frac{2}{d(d-1)}\\
        &=\frac{ 2(k+1) (k+2)}{(d+k)^2 (d+k+1)^2}\cdot (1-f)^2\\
        \E |l|^4&=\E(1-\alpha^2)^2(1-\beta^2)^2\left|\braket{\phi'}{\psi'}\right|^4\\
        &=\frac{(d-1)^2 d^2}{(d+k)^2 (d+k+1)^2}\E \left|\braket{\phi'}{\psi'}\right|^4\\
        &=\frac{2(d-2+f)^2+2(d-2+f^2)}{(d+k)^2 (d+k+1)^2}.
    \end{split}
\end{equation}
Next,
\begin{equation}
    \begin{split}
        \E |g|^2|h|^2&=\E\alpha^2(1-\alpha^2)\beta^2(1-\beta^2) \left|\braket{\phi'}{\psi}\right|^2\left|\braket{\phi}{\psi'}\right|^2\\
        &=\frac{ (k+1)^2}{(d+k)^2 (d+k+1)^2}\cdot (1-f)^2\\
        \E |l|^2|g|^2&=\E |l|^2|h|^2\\
        &=\E(1-\alpha^2)^2\beta^2(1-\beta^2)\left|\braket{\phi'}{\psi'}\right|^2 \left|\braket{\phi'}{\psi}\right|^2\\
        &=\frac{(d-1)^2 d (k+1)}{(d+k)^2 (d+k+1)^2}\E\left|\braket{\phi'}{\psi'}\right|^2 \left|\braket{\phi'}{\psi}\right|^2\\
        &=\frac{ (k+1)}{(d+k)^2 (d+k+1)^2}(1-f) (d+2 f-2)\\
        \E |q|^2|g|^2&=\E |q|^2|h|^2\\
        &=\E\alpha^2(1-\alpha^2)\beta^4 f \left|\braket{\phi'}{\psi}\right|^2\\
        &=\frac{ (k+1)^2 (k+2)}{(d+k)^2 (d+k+1)^2}f(1-f)\\
        \E |q|^2|l|^2&=\E\alpha^2(1-\alpha^2)\beta^2(1-\beta^2) f \left|\braket{\phi'}{\psi'}\right|^2\\
        &=\frac{ (k+1)^2}{(d+k)^2 (d+k+1)^2}f(d-2+f).
    \end{split}
\end{equation}

Then the variance is given by
\begin{equation}
\begin{split}
    \Var(w)&=\frac{(d+k)^4}{k^4}\E\left|\braket{u}{v}\right|^4-\frac{(d+2k)^2}{k^4}-\frac{2(d+2k)f}{k^2}-f^2\\
    &=\frac{(d+k)^2}{(d+k+1)^2}\bigg(\frac{(k+2)^2(k+1)^2}{k^4} f^2 +\frac{ 4(k+1) (k+2)}{k^4}\cdot (1-f)^2\\
    &+\frac{2(d-2+f)^2+2(d-2+f^2)}{k^4}\\
    &+\frac{ 4(k+1)^2}{k^4}\cdot (1-f)^2+\frac{8 (k+1)}{k^4}(1-f) (d+2 f-2)+\frac{ 8(k+1)^2 (k+2)}{k^4}f(1-f)\\
    &+\frac{ 4(k+1)^2}{k^4}f(d-2+f)+\frac{8(k+1)^2(f^2-f)}{k^4}\bigg)-\frac{(d+2k)^2}{k^4}-\frac{2(d+2k)f}{k^2}-f^2.
\end{split}
\end{equation}
To calculate an upper bound of the above exact expression of the variance, note that the sum inside the bracket must be positive, so we can relax the coefficient $\frac{(d+k)^2}{(d+k+1)^2}$ to 1. This gives

\begin{equation}
\begin{split}
    \Var(w)&\leq \frac{(k+2)^2(k+1)^2}{k^4} f^2 +\frac{ 4(k+1) (k+2)}{k^4}\cdot (1-f)^2+\frac{2(d-2+f)^2+2(d-2+f^2)}{k^4}\\
    &+\frac{ 4(k+1)^2}{k^4}\cdot (1-f)^2+\frac{8 (k+1)}{k^4}(1-f) (d+2 f-2)+\frac{ 8(k+1)^2 (k+2)}{k^4}f(1-f)\\
    &+\frac{ 4(k+1)^2}{k^4}f(d-2+f)+\frac{8(k+1)^2(f^2-f)}{k^4}-\frac{(d+2k)^2}{k^4}-\frac{2(d+2k)f}{k^2}-f^2\\
    &=\frac{4 f - 2 f^2}{k}+\frac{2 d f+f^2+4}{k^2}+\frac{4 d+4}{k^3}+\frac{d^2+2 d}{k^4}\\
    &=O\left(\frac{f}{k}+\frac{d f}{k^2}+\frac{1}{k^2}+\frac{d}{k^3}+\frac{d^2}{k^4}\right).
\end{split}
\end{equation}

\section{Proof of Lemma~\ref{lemma:blockdiagonal}}
\label{app:blockdiagonal}
In the following we write $\ket{u}$ as $\ket{0}$ and provide an analysis of Bob's post-measurement state
\begin{equation}
    \rho=\binom{d+k-1}{k}\E_{\phi\sim\mathbb{C}^d}\ketbra{\phi}^{\otimes k}\left|\braket{0}{\phi}\right|^{2k}.
\end{equation}

First, we decompose the integral with respect to random vector $\phi\sim\mathbb{C}^d$ using Lemma~\ref{lemma:decomposition} as 
\begin{equation}
    \ket{\phi}=\alpha e^{i\theta}\ket{0}+\sqrt{1-\alpha^2}\ket{v},
\end{equation}
then we can consider the integral with respect to 3 independent random variable $\alpha,\theta\,\ket{v}$, which becomes
\begin{equation}\label{eq:randomizephase}
    \begin{split}
        \rho&=\binom{d+k-1}{k}\E_{\alpha,\theta,v}\left(\alpha^2\ketbra{0}+\alpha e^{i\theta}\sqrt{1-\alpha^2}\ketbra{0}{v}+\alpha e^{-i\theta}\sqrt{1-\alpha^2}\ketbra{v}{0}+(1-\alpha^2)\ketbra{v}\right)^{\otimes k}\alpha^{2k}\\
        &=\binom{d+k-1}{k}\E_{\alpha,v}\sum_{t=0}^k\alpha^{2(k+t)}(1-\alpha^2)^{k-t}\binom{k}{t}^2\Pi_{\mathrm{sym}}^{d,k}\left(\ketbra{0}^{\otimes t}\otimes \ketbra{v}^{\otimes k-t}\right)\Pi_{\mathrm{sym}}^{d,k}.
    \end{split}
\end{equation}
Here we have performed integral over the random phase $\theta$, which follows from Lemma~\ref{lemma:lowerboundrandomphase}. 

From Eq.~\eqref{eq:randomizephase}, the next step is to perform integral over $v$. Note that here $v$ is a random vector on $\vspan\{\ket{1},\dots,\ket{d-1}\}$. Using Lemma~\ref{lemma:symprojector} we have
\begin{equation}
    \E_{v}\ketbra{v}^{\otimes k-t}=\binom{d+k-t-2}{k-t}^{-1}\Pi_{\mathrm{sym},0}^{d-1,k-t}.
\end{equation}
Here we use $\Pi_{\mathrm{sym},0}^{d-1,k-t}$ to denote the orthogonal projector onto $\vee^{k-t}\vspan\{\ket{1},\dots,\ket{d-1}\}$. This gives
\begin{equation}
    \rho=\binom{d+k-1}{k}\E_{\alpha}\sum_{t=0}^k\alpha^{2(k+t)}(1-\alpha^2)^{k-t}\binom{d+k-t-2}{k-t}^{-1}\binom{k}{t}^2\Pi_{\mathrm{sym}}^{d,k}\left(\ketbra{0}^{\otimes t}\otimes \Pi_{\mathrm{sym},0}^{d-1,k-t}\right)\Pi_{\mathrm{sym}}^{d,k}.
\end{equation}

Next, recall that we have defined the projectors $\Pi^t$ onto the subspaces
\begin{equation}
    W^t=\mathrm{span}\left\{\psym{d}{k}\ket{0}^{\otimes t}\otimes \ket{v}^{\otimes k-t}:\ket{v}\in\vspan\{\ket{1},\dots,\ket{d-1}\}\right\},
\end{equation}
where we have dropped the subscript $u$ as here $\ket{u}=\ket{0}$. We show that this projector can be written as
\begin{equation}
    \Pi^t=\binom{k}{t}\Pi_{\mathrm{sym}}^{d,k}\left(\ketbra{0}^{\otimes t}\otimes \Pi_{\mathrm{sym},0}^{d-1,k-t}\right)\Pi_{\mathrm{sym}}^{d,k}.
\end{equation}
Consider a orthogonal basis of $\ssym{d}{k}=\vspan B$ (Definition~\ref{def:symmetricsubspace}), where
\begin{equation}
    B=\left\{\sum_{i:T(i)=\ell}\ket{i_1,\dots,i_k}:\ell=(\ell_0,\dots,\ell_{d-1}),\ell_j\geq 0,\sum_{j=0}^{d-1}\ell_j=k\right\}.
\end{equation}
Here $\ell_j$ denotes the number of times that $j$ appear in the basis. Consider $\ket{b}\in B$ with pattern $\ell$. Clearly, if $\ell_0\neq t$ then $\Pi^t\ket{b}=0$. If $\ell_0=t$, then
\begin{equation}
    \begin{split}
        \Pi^t\ket{b}&=\binom{k}{t}\Pi_{\mathrm{sym}}^{d,k}\left(\ketbra{0}^{\otimes t}\otimes \Pi_{\mathrm{sym},0}^{d-1,k-t}\right)\ket{b}\\
        &=\binom{k}{t}\Pi_{\mathrm{sym}}^{d,k}\ket{0}^{\otimes t}\otimes \ket{b'}\\
        &=\ket{b},
    \end{split}
\end{equation}
where
\begin{equation}
    \ket{b'}=\sum_{i_1,\dots,i_{k-t}\neq 0:T(0^t,i_1,\dots,i_{k-t})=\ell}\ket{i_1,\dots,i_{k-t}}.
\end{equation}
Therefore, $\Pi^t$ is the orthogonal projector onto 
\begin{equation}
    \vspan\left\{\sum_{i:T(i)=\ell}\ket{i_1,\dots,i_k}:\ell=(\ell_0,\dots,\ell_{d-1}),\ell_j\geq 0,\ell_0=t,\sum_{j=0}^{d-1}\ell_j=k\right\},
\end{equation}
which equals to $W^t$.

Now we can write Bob's state as
\begin{equation}
    \rho=\binom{d+k-1}{k}\E_{\alpha}\sum_{t=0}^k\alpha^{2(k+t)}(1-\alpha^2)^{k-t}\binom{d+k-t-2}{k-t}^{-1}\binom{k}{t}\Pi^t.
\end{equation}
The final step is to calculate the integral with respect to $\alpha$. We have
\begin{equation}
\begin{split}
    \E_{\alpha}\alpha^{2(k+t)}(1-\alpha^2)^{k-t}&=\int_{0}^{1}\alpha^{2(k+t)}(1-\alpha^2)^{k-t}2(d-1)(1-\alpha^2)^{d-2}\alpha \dd\alpha\\
    (x=\alpha^2)&=\int_{0}^{1}(d-1)x^{k+t}(1-x)^{d+k-t-2}\dd x\\
    &=(d-1)\frac{(k+t)!(d+k-t-2)!}{(d+2k-1)!}.
\end{split}
\end{equation}
Here we have used the fact that for $m,n\in\mathbb{Z}_+$, $\int_{0}^{1}x^m (1-x)^n \dd x=\frac{m!n!}{(m+n+1)!}$. This concludes the proof of $\rho=\sum_{t=0}^k \beta_t \Pi^t$, where
\begin{equation}
    \begin{split}
        \beta_t&=\binom{d+k-1}{k}(d-1)\frac{(k+t)!(d+k-t-2)!}{(d+2k-1)!}\binom{d+k-t-2}{k-t}^{-1}\binom{k}{t}\\
        &=\frac{(k+t)(k+t-1)\cdots (t+1)}{(d+2k-1)(d+2k-2)\cdots (d+k)}.
    \end{split}
\end{equation}

\section{Proof of Lemma~\ref{lemma:singlecopyvariance}}
\label{app:singlecopyvariance}
As the variance involves up to the 4th moment of random states, we need the following technical lemma.

\begin{lemma}\label{lemma:randomstatemoments}
Let $A,B,C,D$ be Hermitian matrices. Then
\begin{equation}
    \E_{\psi\sim\mathbb{C}^d}\expval{A}{\psi}\expval{B}{\psi}=\frac{1}{d(d+1)}\left(\Tr(A)\Tr(B)+\Tr(AB)\right),
\end{equation}
and
\begin{equation}
\begin{split}
    \E_{\psi\sim\mathbb{C}^d}\expval{A}{\psi}\expval{B}{\psi}\expval{C}{\psi}&=\frac{1}{d(d+1)(d+2)}\bigg(\Tr(A)\Tr(B)\Tr(C)+\Tr(AB)\Tr(C)\\
    &+\Tr(A)\Tr(BC)+\Tr(C)\Tr(AB)+\Tr(ABC)+\Tr(ACB)\bigg),
\end{split}
\end{equation}
and similarly
\begin{equation}
\begin{split}
    &\E_{\psi\sim\mathbb{C}^d}\expval{A}{\psi}\expval{B}{\psi}\expval{C}{\psi}\expval{D}{\psi}\\
    &=\frac{1}{d(d+1)(d+2)(d+3)}\sum_{\pi\in S_4}\Tr(A\otimes B\otimes C\otimes D\cdot  P_d(\pi)).
\end{split}
\end{equation}
\end{lemma}
\begin{proof}
This directly follows from the characterization of $\psym{d}{k}$ (Lemma~\ref{lemma:symprojector}), which gives
\begin{equation}
    \E_{\psi\sim\mathbb{C}^d}\ketbra{\psi}^{\otimes k}=\frac{1}{(d+k-1)\cdots d}\sum_{\pi\in S_k}P_d(\pi).
\end{equation}
\end{proof}

We start with the second term of Eq.~\eqref{eq:totalvariance}. Using the expression for the variance of collision estimator in Lemma~\ref{lemma:collisionvariance}, we have
\begin{equation}
    \E_U\left[\Var(\tilde{g}(U,S)|U)\right]\leq \E_U\left[\frac{g(U)}{m^2}+\frac{1}{m}\left(\sum_b p_b(U) q_b(U)^2+\sum_b p_b(U)^2 q_b(U)\right)\right].
\end{equation}
We have already shown that $\E g(U)=\frac{1+f}{d+1}$, so the first term in RHS is $O(1/m^2 d)$. For the second term, notice that
\begin{equation}
\begin{split}
    \E_U \sum_{b=0}^{d-1} p_b(U) q_b(U)^2&=\E_U \sum_{b=0}^{d-1}\expval{U\rho U^\dag}{b}\expval{U\sigma U^\dag}{b}^2\\
    &=d\E_{\psi\sim\mathbb{C}^d}\expval{\rho}{\psi}\expval{\sigma}{\psi}^2\\
    &=O(1/d^2),
\end{split}
\end{equation}
where the third line follows from Lemma~\ref{lemma:randomstatemoments}. This implies that
\begin{equation}\label{eq:expectedvarianceoversamples}
    \E_U\left[\Var(\tilde{g}(U,S)|U)\right]=O\left(\frac{1}{m^2 d}+\frac{1}{m d^2}\right).
\end{equation}
Next we consider the first term of Eq.~\eqref{eq:totalvariance}, which is
\begin{equation}
    \begin{split}
        \Var_U(\E[\tilde{g}(U,S)|U])&=\Var_{U\sim\mathbb{U}(d)} g(U)\\
        &=\E_{U\sim\mathbb{U}(d)} g(U)^2-\frac{(1+f)^2}{(d+1)^2}\\
        &=\sum_{b_1,b_2=0}^{d-1}\E_{U\sim\mathbb{U}(d)}\expval{U\rho U^\dag}{b_1}\expval{U\sigma U^\dag}{b_1}\expval{U\rho U^\dag}{b_2}\expval{U\sigma U^\dag}{b_2}-\frac{(1+f)^2}{(d+1)^2}.
    \end{split}
\end{equation}
The above sum consists of two different cases: $b_1=b_2$ and $b_1\neq b_2$. In the second case, the distribution of $\{U^\dag \ket{b_1},U^\dag \ket{b_2}\}$ can be understood as follows: first sample a uniformly random vector $\ket{\psi}\sim\mathbb{C}^d$, then sample a uniformly random vector in  $\mathbb{C}^{d-1}_{\psi_\perp}$, denoted as $\ket{\psi'}\sim\mathbb{C}^{d-1}_{\psi_\perp}$. Recall that $\mathbb{C}^{d-1}_{\psi_\perp}$ is the $d-1$ dimensional subspace of $\mathbb{C}^d$ that is perpendicular to $\ket{\psi}$. Then we have
\begin{equation}\label{eq:varusimplified}
    \begin{split}
        \Var_U(\E[\tilde{g}(U,S)|U])&=d\E_{\psi\sim\mathbb{C}^d}\left(\expval{\rho}{\psi}\expval{\sigma}{\psi}\right)^2\\
        &+d(d-1)\E_{\substack{\psi\sim\mathbb{C}^d\\\psi'\sim\mathbb{C}^{d-1}_{\psi_\perp}} }\expval{\rho}{\psi}\expval{\sigma}{\psi}\expval{\rho}{\psi'}\expval{\sigma}{\psi'}-\frac{(1+f)^2}{(d+1)^2}.
    \end{split}
\end{equation}
To calculate the second term in RHS, we first calculate the integral with respect to $\psi'$ and then calculate the integral with respect to $\psi$, which is
\begin{equation}\label{eq:variancepsiperp}
    \E_{\substack{\psi\sim\mathbb{C}^d\\\psi'\sim\mathbb{C}^{d-1}_{\psi_\perp}} }\expval{\rho}{\psi}\expval{\sigma}{\psi}\expval{\rho}{\psi'}\expval{\sigma}{\psi'}=\E_{\psi\sim\mathbb{C}^d}\left[\expval{\rho}{\psi}\expval{\sigma}{\psi}\E_{\psi'\sim\mathbb{C}^{d-1}_{\psi_\perp} }\expval{\rho}{\psi'}\expval{\sigma}{\psi'}\right].
\end{equation}
We have shown in  Lemma~\ref{lemma:symperpprojector} that
\begin{equation}
    \E_{\psi'\sim \mathbb{C}^{d-1}_{\psi_\perp}}\ketbra{\psi'}^{\otimes 2}=\frac{1}{d(d-1)}\left(\left(I-\ketbra{\psi}\right)^{\otimes 2}+\left(I-\ketbra{\psi}\right)^{\otimes 2}\swapop\left(I-\ketbra{\psi}\right)^{\otimes 2}\right)
\end{equation}
and using this we have
\begin{equation}
\begin{split}
    &d(d-1)\E_{\psi'\sim \mathbb{C}^{d-1}_{\psi_\perp}}\expval{\rho}{\psi'}\expval{\sigma}{\psi'}\\
    &=\left(1-\expval{\rho}{\psi}\right)\left(1-\expval{\sigma}{\psi}\right)+\Tr(\rho(I-\ketbra{\psi})\sigma(I-\ketbra{\psi}))\\
    &=\left(1-\expval{\rho}{\psi}\right)\left(1-\expval{\sigma}{\psi}\right)+f+\expval{\rho}{\psi}\expval{\sigma}{\psi}-\expval{\left(\rho\sigma+\sigma\rho\right)}{\psi}\\
    &=1+f-\expval{\rho}{\psi}-\expval{\sigma}{\psi}+2\expval{\rho}{\psi}\expval{\sigma}{\psi}-\expval{\left(\rho\sigma+\sigma\rho\right)}{\psi}.
\end{split}
\end{equation}
Plugging this into Eq.~\eqref{eq:variancepsiperp}, we have
\begin{equation}
    \begin{split}
        &d(d-1)\E_{\substack{\psi\sim\mathbb{C}^d\\\psi'\sim\mathbb{C}^{d-1}_{\psi_\perp}} }\expval{\rho}{\psi}\expval{\sigma}{\psi}\expval{\rho}{\psi'}\expval{\sigma}{\psi'}\\
        &=\E_{\psi\sim\mathbb{C}^d}\bigg[(1+f)\expval{\rho}{\psi}\expval{\sigma}{\psi}-\expval{\rho}{\psi}^2\expval{\sigma}{\psi}-\expval{\rho}{\psi}\expval{\sigma}{\psi}^2\\
        &+2\expval{\rho}{\psi}^2\expval{\sigma}{\psi}^2-\expval{\rho}{\psi}\expval{\sigma}{\psi}\expval{\left(\rho\sigma+\sigma\rho\right)}{\psi}\bigg]\\
        &=\frac{(1+f)^2}{d (d+1)}+O(1/d^4)
    \end{split}
\end{equation}
where we have ignored the negative terms. Then Eq.~\eqref{eq:varusimplified} becomes
\begin{equation}
    \begin{split}
        \Var_U(\E[\tilde{g}(U,S)|U])&=d\cdot O(1/d^4)
        +\frac{(1+f)^2}{d (d+1)}+O(1/d^4)-\frac{(1+f)^2}{(d+1)^2}\\
        &=O(1/d^3),
    \end{split}
\end{equation}
and the total variance is given by
\begin{equation}
    \Var(\tilde{g}(U,S))=\Var_U(\E[\tilde{g}(U,S)|U])+\E_U[\Var(\tilde{g}(U,S)|U)]=O\left(\frac{1}{d^3}+\frac{1}{m^2 d}+\frac{1}{m d^2}\right).
\end{equation}

To show that the above upper bound is tight, we proceed by calculating the exact expression of the variance when $\rho$ and $\sigma$ are pure states. We start with the following technical lemma.

\begin{lemma}
Let $\rho,\sigma$ be pure states, and $f=\Tr(\rho\sigma)$. Then
\begin{equation}
    \begin{split}
        &\E_{\psi\sim\mathbb{C}^d}\expval{\rho}{\psi}\expval{\sigma}{\psi}=\frac{1+f}{d(d+1)}\\
        &\E_{\psi\sim\mathbb{C}^d}\expval{\rho}{\psi}^2\expval{\sigma}{\psi}=\E_{\psi\sim\mathbb{C}^d}\expval{\rho}{\psi}\expval{\sigma}{\psi}^2=\frac{2+4 f}{d (d+1) (d+2)}\\
        &\E_{\psi\sim\mathbb{C}^d}\expval{\rho}{\psi}\expval{\sigma}{\psi}\expval{\left(\rho\sigma+\sigma\rho\right)}{\psi}=\frac{8 f+4 f^2}{d (d+1) (d+2)}\\
        &\E_{\psi\sim\mathbb{C}^d}\expval{\rho}{\psi}^2\expval{\sigma}{\psi}^2=\frac{4+16 f+4 f^2}{d (d+1) (d+2) (d+3)}.
    \end{split}
\end{equation}
\end{lemma}
\begin{proof}
This follows from a direct calculation with Lemma~\ref{lemma:randomstatemoments}.
\end{proof}

Now we are ready to calculate the variance as follows.
\begin{equation}
\begin{split}
    \Var(\tilde{g}(U,S))&=\Var_U(\E[\tilde{g}(U,S)|U])+\E_U[\Var(\tilde{g}(U,S)|U)]\\
    &=\frac{4+16 f+4 f^2}{ (d+1) (d+2) (d+3)}+\frac{(1+f)^2}{d (d+1)}-\frac{2 (4 f+2)}{d (d+1) (d+2)}\\
    &+\frac{2 \left(4 f^2+16 f+4\right)}{d (d+1) (d+2) (d+3)}-\frac{4 f^2+8 f}{d (d+1) (d+2)}-\frac{(1+f)^2}{(d+1)^2}\\
    &+\E_U\left[\frac{g}{m^2}+\frac{(m-1)}{m^2}\sum_{b=0}^{d-1} p_b q_b^2+\frac{(m-1)}{m^2}\sum_{b=0}^{d-1} p_b^2 q_b-\frac{2m-1}{m^2}g^2\right]\\
    &=\frac{d^2 (1+f)^2-d (6-f) f+d+2 (1-f)^2}{d (d+1)^2 (d+2) (d+3)}\\
    &+\frac{1+f}{m^2(d+1)}+\frac{(m-1)}{m^2}\frac{4+8 f}{(d+1) (d+2)}-\frac{2m-1}{m^2}\frac{d^2 (1+f)^2+5 d (1+f)^2+2 (1-f)^2}{d (d+1) (d+2) (d+3)}.
\end{split}
\end{equation}

\section{Inner product estimation with independent classical shadow}\label{app:independentshadow}
Suppose Alice rotates her state $\rho$ with a random unitary $U$ and measures in the computational basis, obtaining result $x$. Define the classical shadow as the following operator
\begin{equation}
    S=(d+1)U^\dag\ketbra{x} U - I.
\end{equation}
It is shown~\cite{Huang2020predicting} that $\E S=\rho$, where the expectation is over the randomness of $U$ as well as the randomness of quantum measurement. Consider the following algorithm for distributed quantum inner product estimation:
\begin{enumerate}
    \item Alice obtains classical shadow $\{S_1,\dots,S_k\}$ using $k$ copies of her unknown state $\rho$.
    \item Bob obtains classical shadow $\{T_1,\dots,T_k\}$ using $k$ copies of his unknown state $\sigma$.
    \item They compute $g=\Tr(\left(\frac{1}{k}\sum_{i=1}^k S_i\right)\cdot \left(\frac{1}{k}\sum_{j=1}^k T_j\right))$.
\end{enumerate}
Note that each classical shadow estimation uses an independent unitary, so the classical shadows are independent from each other. Also, the above algorithm can be implemented with simultaneous message passing without shared randomness.

\begin{theorem}
For any $\varepsilon>0$, the above algorithm returns an estimate of $f=\Tr(\rho\sigma)$ within $\varepsilon$ additive error with probability at least $2/3$, provided that
\begin{equation}
    k\geq C\cdot \max\left\{\frac{1}{\varepsilon^2},\frac{d}{\varepsilon}\right\}
\end{equation}
for some constant $C>0$ that does not depend on $f$.
\end{theorem}
\begin{proof}

As $\E S_i=\rho$ and $\E T_j=\sigma$, the estimator $g$ satisfies $\E g=f:=\Tr(\rho \sigma)$. It remains to analyze the variance of $g$. We have
\begin{equation}
    \begin{aligned}
    \Var(g)&=\frac{1}{k^4}\E\left(\sum_{i,j=1}^k\Tr(S_i T_j)\right)^2 -f^2\\
    &=\frac{1}{k^4}\E\sum_{i,j,q,l=1}^k\Tr(S_i T_j)\Tr(S_q T_l) -f^2.
    \end{aligned}
\end{equation}
There are four cases:
\begin{enumerate}
    \item $i=q$, $j=l$, there are $k^2$ terms of the form $\E \Tr(ST)^2$.
    \item $i\neq q$, $j\neq l$, this gives $k^2(k-1)^2 f^2$.
    \item $i=q$, $j\neq l$, there are $k^2(k-1)$ terms of the form $\E\Tr(ST_1)\Tr(ST_2)$.
    \item $i\neq q$, $j=l$, there are $k^2(k-1)$ terms of the form $\E\Tr(S_1 T)\Tr(S_2 T)$.
\end{enumerate}
Below we calculate these terms assuming $\rho$ and $\sigma$ are pure states; the results will be upper bounds of the case when $\rho$ and $\sigma$ are mixed states. Let's start with case 3. We have
\begin{equation}
    \begin{aligned}
    \E\Tr(ST_1)\Tr(ST_2)&=\E\Tr(S\sigma)^2\\
    &= \E_U \sum_x\expval{U \rho U^\dag}{x}\left((d+1)\expval{U \sigma U^\dag}{x}-1\right)^2\\
    &=d \E_\psi\expval{\rho}{\psi}\left((d+1)\expval{\sigma}{\psi}-1\right)^2\\
    &=\E_{\psi} d(d+1)^2\expval{\rho}{\psi}\expval{\sigma}{\psi}^2-2d(d+1)\expval{\rho}{\psi}\expval{\sigma}{\psi}+1\\
    &=d(d+1)^2\frac{2 + 4f}{d(d + 1)(d + 2)}-2d(d+1)\frac{1+f}{d(d+1)}+1\\
    &=\frac{d(1+2f)}{d+2}.
    \end{aligned}
\end{equation}
Now the variance can be written as
\begin{equation}
\begin{aligned}
\Var(g)&=\frac{1}{k^2}\E \Tr(ST)^2+\frac{(k-1)^2}{k^2}f^2+\frac{2(k-1)}{k^2}\frac{d(1+2f)}{d+2}-f^2\\
    &= \frac{1}{k^2}\E \Tr(ST)^2+O(1/k).
\end{aligned}
\end{equation}
It remains to calculate $\E \Tr(ST)^2$. We have
\begin{equation}
    \begin{aligned}
    \E \Tr(ST)^2&=\E_{U,V}\sum_{x,y}\expval{U \rho U^\dag}{x}\expval{V \sigma V^\dag}{y}\left((d+1)^2\left|\matrixel{x}{UV^\dag}{y}\right|^2-(d+2)\right)^2\\
    &=d^2\E_{\phi,\psi}\expval{\rho}{\psi}\expval{\sigma}{\phi}\left((d+1)^2\left|\braket{\psi}{\phi}\right|^2-(d+2)\right)^2\\
    &=\E_{\phi,\psi}d^2(d+1)^4\expval{\rho}{\psi}\expval{\sigma}{\phi}\left|\braket{\psi}{\phi}\right|^4\\
    &-2d^2(d+1)^2(d+2)\expval{\rho}{\psi}\expval{\sigma}{\phi}\left|\braket{\psi}{\phi}\right|^2+(d+2)^2\\
    &=\E_{\phi,\psi}d^2(d+1)^4\expval{\rho}{\psi}\expval{\sigma}{\phi}\left|\braket{\psi}{\phi}\right|^4-2(d+2)(d+2+f)+(d+2)^2\\
    &=\frac{(d+1)^2}{(d+2)^2}\left(2d^2+10d+4d f+12+8f\right)-2(d+2)(d+2+f)+(d+2)^2\\
    &=\frac{d^3+2 d^2 f+4 d^2+2 d-4 f-2}{d+2}\\
    &=O(d^2).
    \end{aligned}
\end{equation}
Here, the fifth and sixth lines are calculated using the technique of writing moments of a random pure state as sum of permutations (Lemma~\ref{lemma:randomstatemoments}). This gives 
\begin{equation}
    \Var(g)=O\left(\frac{d^2}{k^2}+\frac{1}{k}\right).
\end{equation}
Note that the result $\E \Tr(ST)^2=O(d^2)$ is consistent with Lemma 6 of \cite{Huang2020predicting}.
\end{proof}

\end{document}